\documentclass[a4paper,10pt]{article}
\usepackage[utf8]{inputenc}
\usepackage{graphicx}
\usepackage{moreverb}
\usepackage{enumitem}
\usepackage{amsmath,amssymb,amsfonts,amsthm,latexsym}
\usepackage[toc,page]{appendix}
\usepackage{fancyhdr}
\usepackage[numbers]{natbib}
\newcommand {\bdm}{\begin{displaymath}}
\newcommand {\edm}{\end{displaymath}}

\newcommand {\R}{\mathbb R}
\newcommand{\N}{\mathbb N}
\newcommand{\p}{\mathbb P}
\newcommand{\F}{\mathbb F}
\newcommand{\h}{\mathcal H}
\newcommand{\E}{\mathbb E}

\newcommand{\f}{\mathcal F}

\newcommand{\g}{\mathcal G}
\newcommand{\G}{\mathbb G}

\newcommand{\B}{\mathcal B}
\newcommand{\T}{\mathcal T}

\DeclareMathOperator*{\esssup}{ess\,sup}

\newtheorem{thm}{Theorem}[section]
\newtheorem{lem}[thm]{Lemma}
\newtheorem{exa}[thm]{Example}
\newtheorem{pro}[thm]{Proposition}
\newtheorem{defn}[thm]{Definition}
\newtheorem{cor}[thm]{Corollary}
\newtheorem{rem}[thm]{Remark}

\newtheorem{assum}[thm]{Assumption}

\title{American Options with Asymmetric Information and Reflected BSDE}
\author{Neda Esmaeeli\\Department of Mathematical Sciences\\ Sharif University of Technology\\ Azadi Avenue\\ Tehran 14588-89694\\ Iran\\E-mail: n\_esmaili@mehr.sharif.ir    \and Peter Imkeller$^{*}$\\ Institut f\"ur Mathematik\\ Humboldt-Universit\"at zu Berlin\\ Unter den Linden 6\\ 10099 Berlin\\Germany\\E-mail: imkeller@math.hu-berlin.de}
\date{}

\begin{document}

\maketitle

\begin{abstract}
We consider an American contingent claim on a financial market where the buyer has additional information. Both agents (seller and buyer)
observe the same prices, while the information available to them may differ due to some extra exogenous knowledge the buyer has.
The buyer's information flow is modeled by an initial enlargement of the reference filtration. It seems natural to investigate the value of the
American contingent claim with asymmetric information. We provide a representation for the cost of the additional information relying on some results
on reflected backward stochastic differential equations (RBSDE). This is done by using an interpretation of prices of American contingent claims with extra information for the buyer by solutions of appropriate
RBSDE.
\end{abstract}



\section{Introduction}

A European contingent claim is a contract on a financial market whose payoff depends on the market state at maturity or exercise time. The problem of valuation and hedging of contingent claims on complete markets, first studied by Black and Scholes \citep{Black}, Merton \citep{Merton1, Merton2},
Harrison and Kreps \citep{Harri1}, Harrison and Pliska \citep{Harri2}, Duffie \citep{Duffi}, and Karatzas \citep{Kara}, among others, can be formulated in terms of backward stochastic differential equations (BSDE). Pricing and hedging on incomplete markets has been investigated by many authors for some decades. We only mention pioneering papers by F\"{o}llmer and Schweizer \citep{Foll}, M\"uller \citep{Muller}, F\"ollmer and Sondermann \citep{Follmer}, Schweizer \citep{Sch}, Sch\"al \citep{Schal}, Bouchaud and Sornette
\citep{Bouch} and El Karoui and Quenez \citep{Nichole2} who were among the first to link this problem to BSDE. BSDE were introduced, on a Brownian filtration, by Bismut \citep{Bismut}.
Pardoux and Peng \citep{Pardoux} proved existence and uniqueness of adapted solutions under suitable square-integrability assumptions for coefficients and  terminal condition.
For some decades, BSDE represent a vibrant field of research, due to its close ties with stochastic control and mathematical finance.

 In contrast to their European counterparts, American contingent claims (ACC), such as American call or put options, can be exercised
at any time before maturity. Ignoring interest rates, it is well known that the value of the process of an
American contingent claim is related to the Snell envelope of the payoff process, i.e. the smallest supermartingale dominating it. The optimal exercise time is given by the hitting time of the payoff process by the Snell envelope. This key observation links optimal stopping problems to reflected backward stochastic differential equations (RBSDE), i.e. BSDE constrained to stay above a given barrier which in the case of the ACC is given by the payoff function. RBSDE in continuous time, the variant related to ACC, were first investigated in El Karoui et al. \citep{Nichole}. In this context the solution process is kept above the reflecting barrier by means of an additional process. As in the classical Skorokhod problem, this process is non-decreasing. The support of the associated positive random measure is included in the set of times at which the solution process touches the barrier.

In this paper, we consider American contingent claims in a scenario in which the buyer has better information than the seller. While the decisions of the latter are based on the public information flow $\F=(\f_{t})_{t\in[0,T]}$, the buyer possesses additional information modeled by some random variable $G$ which is already available initially. So his information evolution is described by the enlarged filtration $\G=(\g_{t})_{t\in[0,T]}$ with $\g_t = \f_t\vee \sigma(G)$.
We study the effect of this additional information on the value and the optimal exercise time of an American contingent claim.
The situation is similar to an insider's optimal investment problem in the simplest possible model, where he aims to maximize expected utility from the terminal value of his portfolio, and his investment decisions are based on the associated larger flow of information.
Pikovsky and Karatzas \citep{Pik} first studied this problem in the framework of an initially enlarged filtration.
Variants of the model were investigated among others by Elliott et al. \citep{Elli}, Grorud and Pontier \citep{Gror1, Gror2}, Amendinger et al. \citep{Amendinger}, or Ankirchner et al. \citep{Ankirchner}.

Building on results about initial enlargements of filtrations by Jacod \citep{jacod2}, in the first part of the paper we reduce the problem to a
standard optimal stopping problem on an enlarged probability space in case $G$ possesses conditional laws with respect to the smaller filtration that are smooth enough (density hypothesis). Under the density hypothesis we write the value function of an American contingent claim obtained with additional information as the value function of a modified American contingent claim on the enlarged space.
To define it as the product of the underlying probability space and the (real) space of possible values of $G$, we give a factorization of $\G$--stopping times in terms of parametrized $\F-$stopping times. This
is a rational choice, since the initial enlargement is related to a measure change on this product space; see for instance Jacod \citep{jacod2} or Amendinger et al. \citep{Amendinger}.

In the second part, following the well known link between optimal stopping problems and RBSDE in El Karoui et al. \citep{Nichole}, on a Brownian basis we define a corresponding
RBSDE on the product space associated to the initial enlargement of the filtration. BSDE for (initially or progressively) enlarged filtrations have been studied by Eyraud-Loisel \citep{Anne} or Kharroubi et al. \citep{Idris}. The approach used in
\citep{Anne} is based on measure changes, which is one, but not the main, tool for our approach. Our treatment of the RBSDE is based on Ito calculus and the canonical decomposition of semimartingales in $\G$.
Extending results in El Karoui et al. \citep{Nichole}, we rewrite the value function of the American contingent claim with asymmetric information
in terms of the solution of the RBSDE on the product space. This provides a solution of the RBSDE with respect to the larger filtration. Possessing additional information, the buyer has a larger value of the expected payoff than the seller.
We study the advantage of the buyer in terms of the solutions of two different RBSDE.

The outline of the paper is the following. After presenting notations and assumptions in Section 2, we introduce the financial market model with asymmetric information.
In Section 3, we factorize $\G$--stopping times as parametrized $\F$--stopping times, and give a formula for the value of an ACC with asymmetric information. We also study the value function for conditional expectations with respect
to the small filtration - an optimal projection problem. Section 4 is concerned with the
link between optimal stopping problems and RBSDE. We recall some results from El Karoui et al. \citep{Nichole} and extend them
to parametrized RBSDE. We define an RBSDE that corresponds to the optimal stopping problem on the product space. By changing variables in the solution of this RBSDE, we obtain an alternative
expression for the value function with additional information in terms of the solution of the RBSDE in the initially enlarged filtration.
In Section 5, we define the cost of additional information by utility indifference. We obtain a formula for the cost in terms of
a difference of solutions of two RBSDE on different spaces. Finally, we compute it in a simple case.

\section{Setup and Preliminaries}

Let $T>0$ represent a finite time.
We consider a filtered probability space $\left(\Omega,\f,\F,\p\right)$, where $ \F = \left(\f_{t}\right)_{t \in [0,T]}$  is the reference filtration satisfying the
usual conditions of right-continuity and completeness. Moreover we assume that $\f_{0}$ is trivial. Equations resp. inequalities involving random variables are usually understood in the almost sure sense.
We consider a random variable $G: \Omega \rightarrow \R$.
Let ${\G}$ be the initial enlargement of $\F$ by $G$, i.e. ${\G} =  \left( {\g}_{t}\right)_{t \in [0,T]}$ where ${\g}_{t} = \f_{t} \vee \sigma(G), t \in [0,T]$.

 We denote by $P^{G}$ the law of $G$ and for $t\in[0,T]$ by $P_{t}^{G}(\omega, du)$ the regular version of the conditional law of $G$ given $\f_{t}$.
Throughout this paper, we will assume that Jacod's density hypothesis (\citep{jacod1}, \citep{jacod2}) stated in the following assumption is satisfied.
\begin{assum}\label{jacod_hypothesis}
 For $t\in[0,T]$, the regular conditional law of $G$ given $\f_{t}$ is equivalent with the law of $G$ for $\p$-almost all $\omega \in \Omega$
 i.e.
 \begin{eqnarray*}
  \p\left[ G \in \cdot |\f_{t}\right] \sim \p\left(G\in\cdot\right),\quad\p \text{-- a.s}.
 \end{eqnarray*}
\end{assum}

 According to \citep{jacod2}, for each $t\in[0,T]$ there exists an $\f\otimes\B(\R)-$measurable version of $\alpha_{t}(u)(\omega):=\frac{dP_{t}^{G}(u,\omega)}{dP^{G}(u)}$  which is strictly positive.
 And for each $u\in\R,$ $\{\alpha_{t}(u)\}_{t\in[0,T]}$ is a martingale w.r.t ${\F}$.
We recall that it is shown in \citep[Proposition 1.10]{Amen01} that the strict positivity of $\alpha$ implies the right continuity of the filtration $\mathbb{G}$. Let $t \in \R^{+}$ and $\mathbb{H}$ a filtration in $\mathcal{F}.$ We denote by $\mathcal{T}_{t,T}\left(\mathbb{H}\right)$ the set of $\mathbb{H}-$stopping times with values in $[t,T]$.
\begin{defn}\label{payoff}
 Consider the following payoff process
 \begin{eqnarray}
 R&=& L1_{[0,T)}+\xi1_{\{T\}},
\end{eqnarray}

where $L$ is an $\F-$adapted real-valued c\`{a}dl\`{a}g process and $\xi$ an $\f_{T}-$measurable random variable, satisfying the integrability condition
\begin{eqnarray}
 \E[\sup_{t\in[0,T]}|L_{t}|+|\xi|]<\infty.
\label{integasssump}\end{eqnarray}

For $t\in[0,T], \tau\in\mathcal{T}_{t,T}\left(\mathbb{F}\right),$
the value function of an American contingent claim is defined by
\begin{eqnarray}\label{equality}
 V_t =  \esssup_{\tau \in \mathcal{T}_{t,T}\left(\mathbb{F}\right)}\E\left[ R(\tau)|\f_{t}\right].
\label{def1}\end{eqnarray}
\end{defn}
$\tau$ is the buyer's stopping time and plays the role of a control tool. We suppose throughout this paper that $0 \leq L_T\leq \xi<+\infty$.

We consider an American contingent claim where, in contrast to the seller, the buyer possesses additional information. This extra information may be based for instance on a good analyst
or better software. The additional information is
described by the random variable we denote by $G$. A natural question one may ask is ''what is the value of an American contingent claim with extra information? ``
Another one addresses the following problem. As the buyer has more information, he has access to a larger set of available stopping times leading to a higher expected payoff. This immediately leads to the
question ''what is the cost of this extra information? ``

A filtration usually encodes a flow of information. So it is natural to model extra information by an enlargement of a
filtration. We will consider an initial enlargement of the reference filtration. This means that we add all the extra information
at initial time to the reference filtration.
As introduced above, $\G = \left(\g_{t}\right)_{t \in [0,T]}$ is the initial enlargement of $\F$ by $G$.
Formally, incorporating extra information leads to working on the following product spaces whose second component is the space of possible values of the additional information given by a real valued random variable. So we consider the probability space $(\widehat{\Omega},\widehat{\f},\widehat{\F},\widehat{\p})$, where
\begin{equation}
\begin{aligned}
 \widehat{\Omega}&:= \Omega\times \R,&\\
 \widehat{\f}_{t}&:=\bigcap_{s>t}\left(\f_{s}\otimes \mathcal{B}(\R\right)), t \in [0,T],&\\
 \widehat{\F}&:=({\widehat{\f}}_{t})_{t\in[0,T]}, \quad \widehat{\f}=\f\otimes\mathcal{B}(\R),&\\
 \widehat{\p}&:= \p\otimes \eta,&
\label{product}\end{aligned}
\end{equation}
where $\eta$ is a probability measure on $(\R,\mathcal{B}(\R))$ playing the role of the law of the additional information. Without loss of generality we may assume that $(\widehat{\Omega},\widehat{\f},\widehat{\p})$ is complete and that $\widehat{\f}_{0}$ contains all $\widehat{\p}-$null sets of $\widehat{\f}$. We denote by $\widehat{\E}$ the expectation w.r.t. $\widehat{\p}$. Taking expectations with respect to $\widehat{\p}$ takes into account averaging over the possible values the additional information can assume, with respect to its law governing the second component in the product space. In other words, for a random variable $X$ defined on $\widehat{\Omega}$ we have
\begin{equation}
\widehat{\E}(X)=\E{\left(\int_{\R}X(u)d\eta(u)\right)}.
\label{rrem}\end{equation}
where $X(u)=X(.,u), u\in\R$.

Due to the definition of the value function of an American contingent claim (\ref{def1}),  our first step on the way to answer the above questions is to study

\begin{equation}
 \esssup_{\tau \in \mathcal{T}_{t,T}\left(\mathbb{G}\right)}\E\left[ R(\tau)|\h_{t}\right],
\label{aim}\end{equation}

 where $\h_{t}=\g_{t}$. We also study the case $\h_{t}=\f_{t}$ which will be seen to be understood as an optimal projection problem.
Our main idea is to look for a suitable representation of $\G-$stopping times as ''parametrized'' $\mathbb{F}-$stopping times, and then reduce the problem to a corresponding problem in
a product filtration which contains the reference filtration. We will answer the first question in this section, while the second one is treated in Section \ref{sec4}.
We denote by $$V^{G}:=\esssup_{\tau \in \mathcal{T}_{0,T}\left(\mathbb{G}\right)}\E\left[ R(\tau)|\g_{0}\right]$$ the value of the American contingent claim with extra information.
We will use the density hypothesis to write this value as the value of an American contingent claim in the product filtration $\widehat{\F}.$ For this purpose, we need some properties
of the filtration $\G.$ We begin with the following remark.
\begin{rem}
 $\g_{0}=\sigma(G).$ This holds true by the fact that $\F$ is right-continuous and $\f_{0}$ is trivial.
\label{r2.1}\end{rem}

It is clear that $V^{G}$ is a $\g_{0}-$measurable random variable. Hence by factorization it is of the form $f(G)$ where $f$ is a real-valued measurable function.

\section{American contingent claims in an initially enlarged filtration}\label{section 2}

In this section, we present a characterization of $\G-$stopping times. We then derive a formula representing the value function of an American contingent claim with extra information. Throughout this section, we work on the probability space $(\widehat{\Omega},\widehat{\f},\widehat{\F},\widehat{\p})$ from (\ref{product}) where $\eta=P^{G}.$

\subsection{Factorization of $\G$-stopping times}
We start with the following proposition.

\begin{pro}\label{mypro1}
  Let $X : \widehat{\Omega}\times\R^{+}\longrightarrow\R$ be an $\widehat{\F}-$adapted process. Then for the random variable $G$, the process
  $X(G): {\Omega}\times\R^{+}\longrightarrow\R$ is $\G-$adapted.
\end{pro}

\begin{proof}
 We define
 \begin{eqnarray*}
  &\bar{G}:&\Omega\longrightarrow\widehat{\Omega}\\
         & &\omega\longmapsto(\omega,G(\omega)).\\
 \end{eqnarray*}
 Then for fixed $t\in[0,T]$, we have for each $B\in\B(\R)$
 \begin{equation}
 \begin{aligned}
  {X_{t}(\omega,G(\omega))}^{-1}(B)&=&{X_{t}(\bar{G}(\omega))}^{-1}(B)\\
  &=&\bar{G}^{-1}(X_{t}^{-1}(B)).
  \end{aligned}
\label{stop}\end{equation}

 Since $X_{t}^{-1}(B)\in\widehat{\f}_{t},$ it is sufficient to prove that $\bar{G}^{-1}(C\times D) \in \mathcal{G}_t$ for $C\in\f_{t}$ and $D\in\B(\R)$.

Indeed, we have
$$\bar{G}^{-1}(C\times D)=C\cap G^{-1}(B)\in\f_{t}\vee \sigma(G)=\g_{t}.$$
\end{proof}

\begin{rem}
With similar arguments, one can show that if $X : \widehat{\Omega}\times\R^{+}\longrightarrow\R$ is an $\widehat{\F}-$progressively measurable (resp. predictable) process, then
 $X(G): {\Omega}\times\R^{+}\longrightarrow\R$ is $\G-$progressively measurable(resp. predictable).
\label{remprog}\end{rem}

The following proposition characterizes $\G-$stopping times in terms of $\widehat{\F}-$stopping times.
\begin{pro}\label{my lemma}
  Let $\tau^{'}: \Omega \rightarrow \R^+$ be a random time.
 $\tau^{'}$  is a $\G-$stopping time if and only if there exists an $\widehat{\F}-$stopping time ${\tau}:\widehat{\Omega} \rightarrow \R^+$ such that $\tau^{'}(\omega) = \tau(\omega,G(\omega))$ for $\p-$a.e. $\omega\in \Omega$.

\end{pro}
\begin{proof}
Suppose first that $\tau$ is an $\widehat{\F}-$stopping time. For $t\in[0,T]$ we have to show that $$\{\tau(\omega,G(\omega))\leq t\}\in\g_{t}.$$
We have with $\bar{G}$ as defined in the previous proposition
\begin{equation}
\begin{aligned}
 \{\tau(\omega,G(\omega))\leq t\}&=&(\tau\circ{\bar G})^{-1}(-\infty,t]\\
 &=&\bar{G}^{-1}(\tau^{-1}(-\infty,t])\in \mathcal{G}_t,
\end{aligned}
\end{equation}
where the last equality follows from the proof of this proposition.

Now to prove the inverse claim, we first show that for every $\G-$predictable set $H$ there exists an $\widehat{\F}-$predictable process $\{J_{t}(\omega,u)\}_{t\in[0,T]}$
 which is measurable in $(t,\omega,u)$ such that $$1_{H}(s,\omega)=J_{s}(\omega,G(\omega)),\ \p-a.s,\ s\in[0,T].$$
 We have $$\g_{t}=\f_{t}\vee\sigma(G)=\sigma(\{F\cap G^{-1}(B)\ : \  F\in\f_{t}, B\in\B(\R)\}),\quad t\in[0,T].$$
From the definition of a predictable $\sigma-$algebra, we get
$$\mathcal{P}(\G)=\sigma(\{(t,\infty)\times(F\cap G^{-1}(B)):F\in\f_{t}, B\in\B(\R), t\in[0,T]\}\cup\{\{0\}\times(F_{0}\cap G^{-1}(B)):F_{0}\in\f_{0}, B\in\B(\R)\}).$$
We start with a set in the generator of $\mathcal{P}(\G).$ So let $t\in[0,T]$, $F\in \f_t$, $B\in \B(\R)$ and suppose that $H=(t,\infty)\times(F\cap G^{-1}(B)).$ Define $$J_{s}(\omega,u):=1_{(t,\infty)\times F\times B}(s,\omega,u),\quad s\in[0,T].$$
 Then $J_{s}(\omega,u)$ is $\widehat{\F}-$measurable and $\widehat{\F}-$predictable because $(t,\infty)\times F\times B$ is an $\widehat{\F}-$predictable set. Moreover, for $(s,\omega)\in[0,T]\times \Omega$ we have
 $$1_{H}(s,\omega)=1_{(t,\infty)\times F}(s,\omega)\cdot1_{(t,\infty)\times B}(s,G(\omega))=J_{s}(\omega,G(\omega)).$$ For
$H=\{0\}\times(F_{0}\cap G^{-1}(B))$ with $F_0\in \f_0, B\in \B(\R)$ we argue similarly. Now define
\begin{eqnarray*}
 \Lambda:=\{H\in\mathcal{P}(\G)\ |\ \exists J:\ \widehat{\F}-\rm{predictable}\ \rm{such}\ \rm{that}\ 1_{H}(t,\omega)=J_{t}(\omega,G(\omega)),\ \p-a.s,\rm{for}\,\, t\in[0,T].\}
\end{eqnarray*}
We know that the generator set of $\mathcal{P}(\G)$ is a subset of $\Lambda.$  Furthermore $\Lambda$ is a $\lambda-$system, so that according to Dynkin's $\pi-\lambda$ theorem, we have $\mathcal{P}(\G)\subseteq\Lambda.$
Now suppose that $\tau^{'}$ is a $\G-$stopping time. Then $[0,\tau^{'}]\in\mathcal{P}(\G)$. So by what has been shown, there exists an $\widehat{\F}-$predictable process $J$ which is measurable
in $(\omega,u)$ such that $1_{[0,\tau^{'}]}(t,\omega)=J_{t}(\omega,G(\omega)),\ \p-a.s,\ t\in[0,T].$ Now define $$\tau(\omega,u):=\inf\{t>0\ : \ J_{t}(\omega,u)=0\}.$$
 The process $J$ is $\widehat{\F}-$predictable so it is $\widehat{\F}-$progressively measurable. Hence by the D\a'{e}but theorem, $\tau$ is an $\widehat{\F}-$stopping time. Moreover, for $\p-a.e\ \omega\in\Omega$ we have $ \tau^{'}(\omega)=\tau(\omega,G(\omega)).$ This completes the proof.

\end{proof}

\begin{cor}
 Let $\tau:\widehat{\Omega} \rightarrow \R^{+}$ be an $\widehat{\F}-$stopping time. Then for every $u\in\R,\ \tau(u)=\tau(\cdot,u)$ is an $\F-$stopping time.
\label{cor1}\end{cor}
\begin{proof}  Let $u_{0}\in\R$ and $t\in[0,T]$. Then
 $$\{\omega\ |\ \tau(\omega,u_{0})\leq t\}\times\{u_{0}\}=\{(\omega,u_{0})\  |\ \tau(\omega,u_{0})\leq t\} \in\widehat{\f}_{t}={\displaystyle\bigcap_{s>t}}(\f_{s}\otimes\B(\R)).$$
Hence $\{\omega\ |\ \tau(\omega,u_{0})\leq t\}\in{\displaystyle\bigcap_{s>t}}\f_{s}=\f_{t}$. Since $u_{0}$ is arbitary the proof is complete.
\end{proof}

\subsection{Value function in an initially enlarged filtration}
We recall a ''parametrized'' version of the conditional expectation.
\begin{lem}
 Let $(U,\mathcal{U})$ be a measurable space and $X:\Omega\times U\rightarrow\R$ be an $\f\otimes\mathcal{U}-$measurable random variable
 satisfying one of the conditions

 \quad

 (1) $X$ is positive,

 \quad

 (2) $\forall u\in U,\ \E[|X(.,u)|]<\infty.$

 \quad

 Then there exists a $\g\otimes\mathcal{U}-$measurable random variable $Y:\Omega\times U\rightarrow\R,$ such that for all $u\in U$
 $$ Y(.,u)=\E[X(.,u)|\g], \quad \p-a.s.$$
\label{yor}\end{lem}
\begin{proof}
 See \citep{yo}, p. 115.
\end{proof}

\begin{rem}
  We denote a random variable $X:\widehat{\Omega}\rightarrow\R,$ by $X(.)$ to emphasize its dependence on a parameter. Obviously we mean $X(u)=X(\omega,u), \omega\in\Omega.$
 \end{rem}

For our next steps we need to introduce the following notation. Recall the payoff process $R$, and set
 \begin{equation}R:{\widehat\Omega}\times\R^{+}\rightarrow\R, (u,t)\mapsto L_{t}\alpha_{t}(u)1_{[0,T[}(t)+\xi\alpha_{T}(u)1_{\{T\}}(t).
\label{newpay} \end{equation}
We denote this new payoff function on the product space with $R$ again. Note that, opposed to the first one, it now acts on two variables.

\begin{rem}
Note that for an $\widehat{\F}-$stopping time $\tau:{\widehat\Omega}\rightarrow\R^{+}$, $R(.,\tau(.)):\widehat{\Omega}\rightarrow\R$ is a positive $\widehat{\f}-$measurable random variable. Since it is a payoff function, Lemma \ref{yor} guarantees the existence of an $\f_{t}\otimes\B(\R)-$measurable version of
 $\E[R(u,\tau(u))|\f_{t}]$ for $u\in\R, t\in[0,T].$
\label{Yor}\end{rem}






\begin{pro}
 Let $t\in[0,T]$. Then the following equation holds
 $$\E{\left[R(u,\tau(u))|\f_{t}\right]}_{u=G}=\widehat{\E}{\left[R(.,\tau(.))|\widehat{\f}_{t}\right]}_{G},\quad \p-a.s.$$

\label{my proposition}\end{pro}

\begin{proof}
  We will show that for every bounded $\f_{t}\otimes\B(\R)-$measurable random variable $K:\Omega\times\R\rightarrow\R$ we have

 \begin{equation}
  \E{\left[\E{\left[R(u,\tau(u))|\f_{t}\right]}_{u=G}K(G)\right]}=\E{\left[\widehat{\E}{\left[R(.,\tau(.))|\widehat{\f}_{t}\right]}_{G}K(G)\right]}.
\label{equn}\end{equation}

 Since both $\E{\left[R(u,\tau(u))|\f_{t}\right]}_{u=G}$ and $\widehat{\E}{\left[R(.,\tau(.))|\widehat{\f}_{t}\right]}_{G}$ are $\g_{t}-$measurable random
 variables, the assertion then follows from (\ref{equn}) and monotone class arguments.\\
 To show (\ref{equn}), note that $K(.)$ and $\alpha_{t}(.)$ are $\f_{t}\otimes\B(\R)-$measurable,
 hence $K(u)$ and $\alpha_{t}(u)$ are $\f_{t}-$measurable for $u\in\R.$ We obtain
 \begin{eqnarray*}
  \E{\left[\E{\left[R(u,\tau(u))|\f_{t}\right]}_{u=G}K(G)\right]}&=&\E{\left[\E{\left[\E{\left[R(u,\tau(u))|\f_{t}\right]_{u=G}K(G)}|\f_{t}\right]}\right]}\\
  &=&\E{\left[\int_{\R}{\E{\left[R(u,\tau(u))|\f_{t}\right]}K(u)\alpha_{t}(u)dP^{G}(u)}\right]}\\
  &=&\E{\left[\int_{\R} R(u,\tau(u))K(u)\alpha_{t}(u)dP^{G}(u)\right]}.
 \end{eqnarray*}
On the other hand,
\begin{eqnarray*}
  \E{\left[\widehat{\E}{\left[R(.,\tau(.))|\widehat{\f}_{t}\right]}_{G}K(G)\right]}&=&\E{\left[\E{\left[\widehat{\E}{\left[R(.,\tau(.))|\widehat{\f}_{t}\right]}_{G}K(G)|\f_{t}\right]}\right]}\\
  &=&\E{\left[\int_{\R}{\widehat{\E}{\left[R(.,\tau(.))|\widehat{\f}_{t}\right]}_{u}K(u)\alpha_{t}(u)dP^{G}(u)}\right]}\\
  &=&\widehat{\E}{\left[\widehat{\E}{\left[R(.,\tau(.))K(.)\alpha_{t}(.)|\widehat{\f}_{t}\right]}\right]}\\
  &=&\E{\left[\int_{\R}{R(u,\tau(u))}K(u)\alpha_{t}(u)dP^{G}(u)\right]}.
  \end{eqnarray*}
The last two equations are satisfied by the definition of $\widehat{\E}$ in {\eqref{rrem}}.
 \end{proof}

\begin{rem}
 Let $t\in[0,T], u\in\R$ and $G=u$ be constant $\p-$a.s. Then from Remark \ref{Yor} we have
 $$\E{\left[R(u,\tau(u))|\f_{t}\right]}=\widehat{\E}{\left[R(.,\tau(.))|\widehat{\f}_{t}\right]}_{u},\quad \widehat{\p}-a.s.$$
\label{remark}\end{rem}

The following result gives a useful clue to calculate conditional expectations with respect to the larger filtration.
\begin{lem}
 Suppose that $X:{\widehat\Omega}\times\R^{+}\rightarrow\R$ is a process, $t\in[0,T]$ and $G:\Omega\rightarrow\R$ a random variable such that $X_{t}(G)$ is $\g_{t}-$measurable and $\p-$integrable. Then for $s\leq t$
 $$\E{\left[X_{t}(G)|\g_{s}\right]}=\frac{1}{\alpha_{s}(G)}\E{\left[X_{t}(u)\alpha_{t}(u)|\f_{s}\right]}_{u=G}.$$
\label{zargari}\end{lem}
\begin{proof}
 See \citep{Zargar}, p. 5.
\end{proof}

\begin{thm}
 Let $t\in[0,T]$. Under Assumption \ref{jacod_hypothesis} on $G$ and the integrability condition (\ref{integasssump}) on $R$ we have for $t\in[0,T]$
 $$V^{G}_{t}:= \esssup_{\tau^{'}\in\T_{t,T}{(\G)}}\E{\left[R(\tau^{'})|\g_{t}\right]}=
 {\frac{1}{\alpha_{t}(G)}{\left(\esssup_{\tau(.)\in\T_{t,T}{(\widehat{\F})}}\widehat{\E}\left[R(.,\tau(.))|\widehat{\f}_{t}\right]\right)}_{G}}.$$
\label{thm1}\end{thm}

\begin{proof}
 Let $\tau^{'}\in\T_{t,T}(\G)$. From Proposition \ref{my lemma} and Lemma \ref{zargari}, we have

 $$\E{\left[R(\tau^{'})|\g_{t}\right]}=\E{\left[R(\tau(G))|\g_{t}\right]}=\frac{1}{\alpha_{t}(G)}\E{\left[R(\tau(u))\alpha_{T}(u)|\f_{t}\right]}_{u=G},$$
 where $\tau(.)\in\T_{t,T}{(\widehat{\F})}$.

\quad

 From Corollary \ref{cor1} for $u\in\R,\ \tau(u)$ is an $\F-$stopping time. So by using iterated conditional expectations
 and the martingale property of $(\alpha_{t}(u))_{t\in[0,T]}$ w.r.t $\F$, we get

 \begin{eqnarray*}
  {\E{\left[R(\tau(u))\alpha_{T}(u)|\f_{t}\right]}}&=&{\E{\left[\E{\left[(L_{\tau(u)}1_{[0,T[}(\tau(u))+\xi 1_{\{T\}}(\tau(u)))\alpha_{T}(u)|\f_{\tau(u)}\right]}|\f_{t}\right]}}\\
    &=& {\E{\left[L_{\tau(u)}\alpha_{\tau(u)}(u)1_{[0,T[}(\tau(u)) +\xi 1_{\{T\}}(\tau(u))\alpha_{T}(u)|\f_{t}\right]}}\\
    &=&{\E{\left[R(u,\tau(u))|\f_{t}\right]}}.
 \end{eqnarray*}
 Thus we have
 \begin{eqnarray*}
  {\esssup_{\tau^{'}\in\T_{t,T}{(\G)}}\E{\left[R(\tau^{'})|\g_{t}\right]}}&=&{\esssup_{\tau(.)\in\T_{t,T}{(\widehat{\F})}}\E{\left[R(\tau(G))|\g_{t}\right]}}\\
  &=&{\esssup_{\tau(.)\in\T_{t,T}{(\widehat{\F})}}\frac{1}{\alpha_{t}(G)}\E{\left[R(\tau(u))\alpha_{T}(u)|\f_{t}\right]}_{u=G}}\\
 &=&{\frac{1}{\alpha_{t}(G)}\esssup_{\tau(.)\in\T_{t,T}{(\widehat{\F})}}\left(\E{\left[R(u,\tau(u))|\f_{t}\right]}\right)_{u=G}}\\
 &=&{\frac{1}{\alpha_{t}(G)}\esssup_{\tau(.)\in\T_{t,T}{(\widehat{\F})}}{\left(\widehat{\E}\left[R(.,\tau(.))|\widehat{\f}_{t}\right]\right)}_{G}}.\\
 \end{eqnarray*}





The last equality comes from  Proposition \ref{my proposition}. Moreover,
 $\widehat{\E}{[R(.,\tau(.))|{\widehat{\f}}_{t}]}$ is measurable in $(\omega,u)$, and the essential supremum of
 a measurable family ${\{\widehat{\E}{[R(.,\tau(.))|{\widehat{\f}}_{t}]};\ {\tau(.)\in\T_{t,T}{(\widehat{\F})}}\}}$ is again measurable in $(\omega,u).$ Therefore, we have

 \begin{eqnarray*}
  {\esssup_{\tau(.)\in\T_{t,T}{(\widehat{\F})}}\widehat{\E}{\left[R(.,\tau(.))|\widehat{\f}_{t}\right]}_{u}}
  ={\left(\esssup_{\tau(.)\in\T_{t,T}{(\widehat{\F})}}\widehat{\E}{\left[R(.,\tau(.))|\widehat{\f}_{t}\right]}\right)_{u}},\quad \p-a.s
\end{eqnarray*}
 and this still holds $\p-$a.s. if we replace $u$ by $G(\cdot)$.

All in all, we obtain as claimed
 \begin{eqnarray*}
  {\esssup_{\tau^{'}\in\T_{t,T}{(\G)}}\E{\left[R(\tau^{'})|\g_{t}\right]}}&=&{\frac{1}{\alpha_{t}(G)}\esssup_{\tau(.)\in\T_{t,T}{(\widehat{\F})}}{\left(\widehat{\E}\left[R(.,\tau(.))|\widehat{\f}_{t}\right]\right)}_{G}}\\
  &=&{\frac{1}{\alpha_{t}(G)}{\left(\esssup_{\tau(.)\in\T_{t,T}{(\widehat{\F})}}\widehat{\E}\left[R(.,\tau(.))|\widehat{\f}_{t}\right]\right)}_{G}}.
 \end{eqnarray*}

\end{proof}

From Neveu \citep{Nvu}, it is known that the essential supremum of a family $\cal{A}$ of non negative random variables is a well defined
almost surely unique random variable. Moreover, if $\cal{A}$ is \emph{directed above}, i.e. $a\vee a^{'}\in\cal{A}$ for $a$ and $a^{'}\in\cal{A}$, then there exists a sequence $(a_{n})_{n\in\N}$ in $\cal{A}$ such that $a_{n}\uparrow(\esssup\cal{A})$ as $n\rightarrow\infty$. See Proposition (VI-1.1) in \citep{Nvu} for a complete proof.

\begin{pro}
 There exists a sequence of stopping times $(\tau_{n})_{n\in\mathbb{N}}$ with $\tau_{n}$ in $\T_{t,T}$ for $n\in\mathbb{N}$ such that the sequence
 $(\E[R(\tau_{n})|\f_{t}])_{n\in\mathbb{N}}$ is increasing and such that $$V_{t}=\lim_{n\rightarrow\infty}\uparrow\E[R(\tau_{n})|\f_{t}],\quad \p-a.s.$$
\label{prokomaki}\end{pro}
\begin{proof}
It is sufficient to show that the set $\{\E[R(\tau)|\f_{t}];\ \tau\in\T_{t,T}\}$ is directed above. Then the result follows from
known results on the essential supremum by Neveu \citep{Nvu}. See Kobylanski and Quenez \citep{Kob} for details of the proof and a complete discussion for the general case where a deterministic time $t$ is replaced by a stopping time in $\T_{0,T}$.
\end{proof}

\begin{thm}
 Let $t\in[0,T]$. Under Assumption \ref{jacod_hypothesis} on $G$ and the integrability condition (\ref{integasssump}) on $R$, we have for $t\in[0,T]$
 $$\esssup_{\tau^{'}\in\T_{t,T}{(\G)}}\E{\left[R(\tau^{'})|\f_{t}\right]}=
 {\int_{\R}{(\esssup_{\tau(.)\in\T_{t,T}{(\widehat{\F})}}\widehat{\E}{\left[R(.,\tau(.))|{\widehat{\f}}_{t}\right]})_{u}}dP^{G}(u)}.$$
\label{thm2}\end{thm}
\begin{proof}
 Let $\tau^{'}\in\T_{t,T}(\G)$. By Proposition \ref{my lemma} there exists an $\widehat{\F}-$stopping time $\tau(.)$ such that $\tau^{'}=\tau(G),\ \p-a.s.$ We therefore have
 $$\E{\left[R(\tau^{'})|\f_{t}\right]}=\E{\left[R(\tau(G))|\f_{t}\right]}.$$ By using the conditional law of $G$ given $\f_{t}$ we get
 \begin{eqnarray*}
  \E{\left[R(\tau(G))|\f_{t}\right]}&=&\E{\left[\E{\left[R(\tau(G))|\f_{T}\right]}|\f_{t}\right]}\\
  &=&{\E{\left[\int_{\R}{R(\tau(u))\alpha_{T}(u)dP^{G}(u)}|\f_{t}\right]}}\\
  &=&{\int_{\R}{\E{\left[R(\tau(u))\alpha_{T}(u)|\f_{t}\right]}}dP^{G}(u)}\\
  &=&{\int_{\R}{\E{\left[\E{\left[R(\tau(u))\alpha_{T}(u)|\f_{\tau(u)}\right]}|\f_{t}\right]}}dP^{G}(u)}\\
  &=&{\int_{\R}{\E{\left[L_{\tau(u)}\alpha_{\tau(u)}(u)1_{[0,T[}(\tau(u))
  +\xi\alpha_{T}(u)1_{\{T\}}(\tau(u))|\f_{t}\right]}}dP^{G}(u)}\\
  &=&{\int_{\R}{\E{\left[R(u,\tau(u))|\f_{t}\right]}}dP^{G}(u)}.
 \end{eqnarray*}
Here we use the martingale property of $(\alpha_{t}(u))_{t\in[0,T]}$ w.r.t $\F$.

 From Remark \ref{remark} we further deduce

 \begin{eqnarray*}
  {\esssup_{\tau^{'}\in\T_{t,T}{(\G)}}\E{\left[R(\tau^{'})|\f_{t}\right]}}&=&{ \esssup_{\tau(.)\in\T_{t,T}{(\widehat{\F})}}\E{\left[R(\tau(G))|\f_{t}\right]}}\\
    &=& {\esssup_{\tau(.)\in\T_{t,T}{(\widehat{\F})}}\int_{\R}{\E{\left[R(u,\tau(u))|\f_{t}\right]}}dP^{G}(u)}\\
    &=& {\esssup_{\tau(.)\in\T_{t,T}{(\widehat{\F})}}\int_{\R}{\widehat{\E}{\left[R(.,\tau(.))|\widehat{\f}_{t}\right]_{u}}}dP^{G}(u)}\\
    &=&{\int_{\R}{\esssup_{\tau(.)\in\T_{t,T}({\widehat{\F})}}\widehat\E{\left[R(.,\tau(.))|\widehat{\f}_{t}\right]}_{u}}dP^{G}(u)},\ \quad\p-a.s.
 \end{eqnarray*}
 To show the last equation we need to prove
 $$\int_{\R}{\esssup_{\tau(.)\in\T_{t,T}{(\widehat{\F})}}\widehat{\E}{\left[R(.,\tau(.))|\widehat{\f}_{t}\right]}_{u}}dP^{G}(u) \leq \esssup_{\tau(.)\in\T_{t,T}{(\widehat{\F})}}\int_{\R}{\widehat{\E}{\left[R(.,\tau(.))|\widehat{\f}_{t}\right]_{u}}}dP^{G}(u),\ \quad\p-a.s.,$$
 the reverse inequality being standard. The measurability of the family ${\{\widehat{\E}{[R(.,\tau(.))|{\widehat{\f}}_{t}]};\ {\tau(.)\in\T_{t,T}{(\widehat{\F})}}\}}$ in $(\omega,u)$
  implies $${\esssup_{\tau(.)\in\T_{t,T}{(\widehat{\F})}}\widehat{\E}{\left[R(.,\tau(.))|{\widehat{\f}}_{t}\right]}_{u}}={(\esssup_{\tau(.)\in\T_{t,T}{(\widehat{\F})}}\widehat{\E}{\left[R(.,\tau(.))|{\widehat{\f}}_{t}\right]})_{u}}.$$
  From Proposition \ref{prokomaki}, there exists $\tau_{n}(.)\in\T_{t,T}(\widehat{\F})$ such that $\widehat{\p}-a.e.$ we have
 $$ \esssup_{\tau(.)\in\T_{t,T}{(\widehat{\F})}}\widehat{\E}{\left[R(.,\tau(.))|
 \widehat{\f}_{t}\right]} = \lim_{n\to \infty}\widehat{\E}{\left[R(.,\tau_{n}(.))|
 \widehat{\f}_{t}\right]}.$$
Therefore by dominated convergence

 \begin{eqnarray*}
 {\int_{\R}{\esssup_{\tau(.)\in\T_{t,T}{(\widehat{\F})}}\widehat{\E}{\left[R(.,\tau(.))|\widehat{\f}_{t}\right]}_{u}}dP^{G}(u)}
  & = & {\lim_{n\to\infty} \int_{\R}{\widehat{\E}{\left[R(.,\tau_{n}(.))|\widehat{\f}_{t}\right]_{u}}}dP^{G}(u) }\\
    &= & {\esssup_{\tau(.)\in\T_{t,T}{(\widehat{\F})}}\int_{\R}{\widehat{\E}{\left[R(.,\tau(.))|\widehat{\f}_{t}\right]_{u}}}dP^{G}(u)},\quad {\p}-a.s.
 \end{eqnarray*}

This finally allows us to deduce
 \begin{eqnarray*}
 {\esssup_{\tau^{'}\in\T_{t,T}{(\G)}}\E{\left[R(\tau^{'})|\f_{t}\right]}}&=&
 {\int_{\R}{\esssup_{\tau(.)\in\T_{t,T}{(\widehat{\F})}}\widehat{\E}{\left[R(.,\tau(.))|{\widehat{\f}}_{t}\right]}_{u}}dP^{G}(u)}\\
  &=&{\int_{\R}{(\esssup_{\tau(.)\in\T_{t,T}{(\widehat{\F})}}\widehat{\E}{\left[R(.,\tau(.))|{\widehat{\f}}_{t}\right]})_{u}}dP^{G}(u)}.
 \end{eqnarray*}

\end{proof}
For both $\h_{t}=\f_{t}\ {\rm or}\ \g_{t}$, we could calculate  the optimal expected payoff (\ref{aim}) based on a value function of a new optimal stopping
problem in the product space. Since optimal stopping problems and reflected BSDE are known to be connected via the Snell envelope, it seems natural to look for the corresponding RBSDE in the product space. This will lead us to consider \emph{parametrized} RBSDE, where the parameter is given by the possible values of a random variable $G$ initially enlarging an underlying filtration. It will be of independent interest to investigate such parametrized RBSDE. This is the goal of the following section. Since the martingale representation property plays an important role for RBSDE, we need to suppose that the reference filtration $\F$ is the natural filtration of a Brownian motion.

\section{RBSDE in an initially enlarged filtration}

\subsection{Basic notions}
Reflected BSDE (RBSDE) were studied by El Karoui et al. \citep{Nichole} on a Brownian basis. Solution processes of such equations
are constrained to keep above a given process called obstacle or barrier. Our work generalizes \citep{Nichole} to the setting of parametrized RBSDE where the reference filtration is the natural filtration of a Brownian motion.

Let $B=(B_{t})_{0\leq t\leq T}$ be a one-dimensional Brownian motion defined on a probability space $(\Omega,\f,\p)$ and  $\F=(\f_{t})_{0\leq t\leq T}$ be the natural filtration of $B$, which satisfies the usual conditions of completion and right continuity.
Denote
\begin{eqnarray*}
\mathcal{L}^2&=&\{X: X\,\,\f_{T}-\mbox{measurable random variable,}\,\, \E(|X|^{2}) < \infty\},\\
\mathcal{H}^2&=&\{X: X = (X_t)_{0\le t\le T} \,\,\mbox{continuous predictable process,}\,\, \E\displaystyle\int_0^T|X_t|^2 dt < \infty\},\\
\mathcal{S}^2&=&\{X: X = (X_t)_{0\le t\le T}\,\, \mbox{continuous predictable process,}\,\,\E(\displaystyle{\sup_{0\le t\le T}} |X_t|^2)<\infty\},\\
\mathcal{I}^{\ }&=&\{K: K = (K_t)_{0\le t\le T}\,\, \mbox{non-decreasing continuous process},\,\, K_{0}=0,\,\, K_{T}\in\mathcal{L}^2\}.
\end{eqnarray*}

As in El Karoui et al. \citep{Nichole} consider a triplet of standard parameters $(\xi, f, L)$ satisfying the following conditions

\quad

\quad(i) \label{shart}$\xi\in \mathcal{L}^2;$

\quad (ii) $f:\Omega\times[0,T]\times\R\times\R\longrightarrow\R$ is such that $f(\cdot,\cdot,y,z)$ is predictable, $\E{\left[\int_{0}^{T}f^{2}(\cdot, t, 0, 0)dt\right]}<\infty,$ and that it is globally Lipschitz continuous in $(y,z)$ for fixed $(\omega,t)\in \Omega\times[0,T]$;


\quad(iii) $L \in \mathcal{S}^2.$

\quad

$\xi$ is called \emph{terminal variable}, $f$ \emph{driver} and $L$ \emph{barrier process}. We shall always assume that $L_{T}\leq\xi.$
A triplet $(Y, Z, K)\in \mathcal{S}^2\times\mathcal{H}^2\times\mathcal{I}$ is a solution of the reflected backward stochastic differential equation (RBSDE) associated with $(\xi, f, L)$ if
it satisfies the following equations resp. inequalities for any $t\in[0,T]$

\begin{equation}
\left\{
\begin{aligned}
Y_{t}&=\xi+\int_{t}^{T}{f(s,Y_s,Z_s)ds}+K_{T}-K_{t}-{\int_{t}^{T}{Z_{s}dB_{s}}},\\
Y_{t}&\geq L_{t},\ \ \ \int_{0}^{T}{(Y_{t}-L_{t})dK_{t}}=0.\\
\end{aligned}
\right.
\label{BSDE}\end{equation}

$K$ controls $Y$ to stay above
the barrier $L$. The condition $\int_{0}^{T}{(Y_{t}-L_{t})\,\,dK_{t}}=0$ which is known as the Skorokhod condition guarantees that the process $K$ acts in a minimal fashion.

If the standard triplet satisfies (i)-(iii), there exists a unique solution of (\ref{BSDE}) (see El Karoui et al. \citep{Nichole}). In case the barrier $L$ is just optional and right upper semicontinuous, the existence of a unique solution of the RBSDE is shown in Grigorova et al. \citep{grigorova}. The component $Y$ is cadlag in this case.

\begin{rem}
If $f$ does not depend on $y$ and $z$,  condition (ii) can be simplified to

\quad(ii*) $f:\Omega\times[0,T]\longrightarrow\R$ is a predictable process s.t. ${\E}{\left[\int_{0}^{T}f^{2}_{t}dt\right]}<\infty.$
\end{rem}

\subsection{RBSDE and optimal stopping problems}\label{secjun}

Snell envelopes provide the well known link between value functions of optimal stopping problems and solutions of corresponding RBSDE (see for example El Karoui et al in \citep{Nichole}).  We shall extend this link to the framework of parametrized RBSDE defined on the product space. We start by recalling some basic facts from the classical theory.

\quad

\begin{pro}
 Let $(Y,Z,K)$ be the solution of the RBSDE (\ref{BSDE}).
 Then $$Y_{t}=\esssup_{{\tau\in\mathcal{T}_{t,T}\left(\mathbb{F}\right)}}\E{\left[\int_{t}^{\tau}{f(s,Y_s,Z_s)ds+L_{\tau}1_{[0,T[}(\tau)+
 \xi1_{\{T\}}}(\tau)|\f_{t}\right]},$$
 where $\mathcal{T}_{t,T}\left(\mathbb{F}\right)$ is the set of all $\F-$stopping times with values in $[t,T].$
\label{avali}\end{pro}

\begin{proof}
 See \citep{Nichole}.
\end{proof}

\begin{pro}
Suppose that $f = (f_t)_{0\le t\le T}$ is an ${\F}-$progressively measurable process that does not depend on $y$ and $z$. Under assumptions (i), (ii*), and (iii), the RBSDE (\ref{BSDE}) with driver $f$ has a unique solution
$\{(Y_t,Z_t,K_t); \ {0 \leq t \leq T}\}$. 
\label{dovomi}\end{pro}

\begin{proof}
 See \citep{Nichole}.
\end{proof}

It is clear from the preceding propositions that in case $f$ does not depend on $y, z$, the link between RBSDE and optimal stopping problems via Snell envelope becomes very explicit. This is stated in the following proposition that is mentioned in \citep{Nichole}.

\begin{pro}
Suppose that $f$ is an ${\F}-$progressively measurable process that does not depend on $y$ and $z$. Under the assumptions (i), (ii*), and (iii),
$Y+\int_{0}^\cdot{f_s ds}$ is the value function of an optimal stopping problem with payoff $$\int_{0}^\cdot f_s ds+L1_{[0,T[} +\xi1_{\{T\}},$$ where  $Y$ is the first component of the solution triplet of the RBSDE (\ref{BSDE})  with coefficient $f$. 
Furthermore, for $t\in[0,T]$ the stopping time
$\tau^{*}=\inf{\{s\in[t,T] : Y_s=L_s\}}\wedge{T}$ is optimal, in the sense $$Y_{t}=\E{\left[\int_{t}^{\tau^{*}}{f_s ds+L_{\tau^{*}}1_{[0,T[}( \tau^{*})+\xi1_{\{T\}}(\tau^{*})}|\f_{t}\right]}.$$
\label{optimal}\end{pro}

\begin{rem}
 If $f\equiv0$ then $Y$, the first component of the solution of RBSDE (\ref{BSDE}), is the value function of the American contingent claim with payoff $L_{t}1_{[0,T[}(t)+\xi1_{\{T\}}(t)$ and $\tau^{*}$ is
 the optimal stopping time for the buyer after time $t$.
\label{rem1}\end{rem}

In the sequel, we suppose that $(\Omega,\f,\F,\p)$ is the filtered probability space carrying a one-dimensional Brownian motion $B$, and $\F=(\f_{t})_{0\leq t\leq T}$ is the Brownian standard filtration.

\subsection{Parametrized RBSDE}
 Recall our product space $(\widehat{\Omega},{\widehat{\f}},{\widehat{\F}},\widehat{\p})$ from (\ref{product}). In order to obtain solutions of RBSDE in the initially enlarged filtration in the following section, the statement of problems with initial enlargements in the framework of product spaces now leads us to consider RBSDE in such product spaces. As the main ingredient for obtaining solutions of RBSDE, we need a martingale representation theorem in this setting. For this purpose, some preparations are needed.
 \begin{rem}
If for a random variable $X:\widehat{\Omega}\rightarrow\R,$ we write $X(.)$, the $\cdot$ stands for the parameter $u\in\R.$
 \end{rem}

 \begin{pro}
  Suppose that $M:\widehat{\Omega}\times[0,T]\rightarrow\R$ is an $\widehat{\f}-$measurable function such that for each $u\in\R,$ $\{M_{t}(u)\}_{t\in[0,T]}$ is a martingale w.r.t $\F$
  and $\int_{\R}\E[|M_{t}(u)|]d\eta(u)<+\infty.$ Then
  $\{M_{t}(.)\}_{t\in[0,T]}$ is a martingale w.r.t $\widehat{\F}.$
 \label{martingale} \end{pro}
 \begin{proof}
For $t\in[0,T]$ we have from Fubini's theorem
 \begin{eqnarray*}
  \widehat{\E}[|M_{t}(.)|]=\int_{\R}\E[|M_{t}(u)|]d\eta(u)<+\infty.
 \end{eqnarray*}

Suppose that $s\leq t$, $C\in{\f}_{s}$, and $D\in\B(\R)$. From Fubini's theorem and martingale property of $\{M_{t}(u)\}_{t\in[0,T]}$ w.r.t $\F$ we have $\widehat{\p}-$a.s.
\begin{equation}
\begin{aligned}
   {\widehat\E[M_{t}(.)1_{C\times D}(.)]}&=\ {\E[\int_{\R}M_{t}(u)1_{C}1_{D}(u)d\eta(u)]}&\\
   &=\ {\int_{\R}\E[M_{t}(u)1_{C}]1_{D}(u)d\eta(u)}&\\
   &=\ {\int_{\R}\E[M_{s}(u)1_{C}]1_{D}(u)d\eta(u)}&\\
   &=\ {\E[\int_{\R}M_{s}(u)1_{C}1_{D}(u)d\eta(u)]}&\\
   &=\ {\widehat\E[M_{s}(.)1_{C\times D}(.)]}.&
 \label{mart}\end{aligned}
 \end{equation}
Now define $$E:=\{A\in\widehat{\f}_{s}\ :\ \widehat\E[M_{t}(.)1_{A}(.)]=\widehat\E[M_{s}(.)1_{A}(.)],\ \widehat{\p}-a.s.\}$$ and
$$H:=\{C\times D;\ C\in{\f}_{s},D\in\B(\R)\}.$$ From (\ref{mart}), we have $H\subseteq E$. Moreover $H$ is a $\pi-$system
and $E$ is a $\lambda-$system so by the Dynkin's $\pi-\lambda$ theorem, we have
$$\forall A\in\widehat{\f}_{s},\ \ \widehat\E[M_{t}(.)1_{A}(.)]=\widehat\E[M_{s}(.)1_{A}(.)].$$

\end{proof}

\begin{cor}
We define $\widehat{B}:\widehat{\Omega}\times[0,T]\rightarrow\R$ by $\widehat{B}_{t}(\omega,u):=B_{t}(\omega)$, where $B=(B_{t})_{t\in[0,T]}$ is a Brownian motion w.r.t.\, $\F$. Then from the above proposition, $\widehat{B}(.)$ is a Brownian motion w.r.t $\widehat{\F}$.
\label{mre}\end{cor}

\begin{pro}
Let $\widehat{X}:\widehat{\Omega}\times[0,T]\rightarrow\R$ and ${X}:{\Omega}\times[0,T]\rightarrow\R$ be two stochastic processes  such that for each $t\in[0,T],\ \widehat{X}_{t}(\omega,u)=X_{t}(\omega),\ u\in\R.$ Then we have $\sigma(\widehat{X}_{s}(.),\ 0\leq s\leq t)=\sigma(X_{s},\ 0\leq s\leq t)\otimes\{\emptyset, \R\},\,t\in[0,T].$
\label{brofil}\end{pro}
\begin{proof}
It is clear that $\sigma({X}_s, 0\le s\le t)\otimes\{\emptyset, \mathbb{R}\}$ is contained in the natural filtration of $\widehat{X}(.)$ on $[0,t]$. On the other hand, since $\widehat{X}$ is constant in the second variable, the natural filtration of $\widehat{X}(.)$ on $[0,t]$ is also contained in $\sigma({X}_s, 0\le s\le t)\otimes\{\emptyset, \mathbb{R}\}$. 
\end{proof}
\begin{cor}
Proposition \ref{brofil} implies that for $\widehat{B}(.)$ defined in Corolary \ref{mre}, we have
$$\sigma(\widehat{B}_{s}(.),\ 0\leq s\leq t)=\sigma(B_{s},\ 0\leq s\leq t)\otimes\{\emptyset, \R\},\ t\in[0,T].$$
Furtheremore, since $\F=(\f_t)_{0\leq t\leq T}$ is the natural filtration generated by $B$, then
$$\sigma(\widehat{B}_{s}(.),\ 0\leq s\leq t)=\f_{t}\otimes\{\emptyset, \R\},\  t\in[0,T].$$
\label{brofi2}\end{cor}
The proposition above implies that the natural filtration generated by $\widehat{B}(.)$ is a subset of the product filtration $\widehat{\F}=(\widehat{\f}_t)_{0\leq t\leq T}$ given by $\widehat{\f}_{t}=\bigcap_{s>t}\left(\f_{s}\otimes \mathcal{B}(\R\right)),\,t\in[0,T]$. So it is not clear a priori that the martingale representation property can be extended to the product space. However, a simple direct argument making use of the product structure will prove that the martingale representation theorem from the first factor extends to the whole space. For more details we need the following preliminaries.

\quad

\begin{cor}
  Let $X:\widehat{\Omega}\times[0,T]\rightarrow\R$ be an $\widehat{\f}\otimes\B([0,T])-$measurable function such that $\int_{0}^{T}X_{s}(u)dB_{s}$ is defined for $u\in\R$, then
 $$\widehat{\E}(\int_{0}^{T}X_{s}(.)dB_{s})^{2}=\widehat{\E}(\int_{0}^{T}X^{2}_{s}(.)d{s}).$$
\label{cor}\end{cor}
\begin{proof}
 It can be easily deduced from the definition of $\widehat{\E}$, Fubini's theorem, and Ito's isometry that
 \begin{eqnarray*}
  \widehat{\E}(\int_{0}^{T}X_{s}(.)dB_{s})^{2}&=&{\E}[\int_{\R}{(\int_{0}^{T}X_{s}(u)dB_{s})^{2}}d\eta(u)]\\
  &=&\int_{\R}{\E(\int_{0}^{T}X_{s}(u)dB_{s})^{2}}d\eta(u)\\
  &=&\int_{\R}{\E(\int_{0}^{T}X^{2}_{s}(u)ds)}d\eta(u)\\
  &=&\widehat{\E}(\int_{0}^{T}X^{2}_{s}(.)ds).
 \end{eqnarray*}
\end{proof}
We introduce some auxiliary spaces on $(\widehat{\Omega},{\widehat{\f}}_{T},\widehat{\p})$. Let $L^{2}$ be the space of $\widehat{\f}_{T}-$measurable random variables
 $X$ that are square integrable i.e. ${\widehat{\E}}(|X|^{2})<+\infty$. We denote by $BL^{2}$ the subspace of $L^{2}$ consisting of bounded elements and by $SBL^{2}$ the subspace of $BL^{2}$ composed of linear combinations of
 random variables of the form $H(\omega)K(u)$ where $H$ is $\f_{T}-$measurable and bounded and $K$ is $B(\R)-$measurable and bounded.
\begin{thm}
 Let $M : \widehat{\Omega}\rightarrow\R$ be an ${\widehat{\f}}_{T}-$measurable random variable such that ${\widehat{\E}}(|M(.)|^{2})<+\infty$. Then there exists a unique $\widehat{\f}-$measurable function $\widehat{Z} : \widehat{\Omega}\times[0,T]\rightarrow\R,$
 which is predictable w.r.t ${\widehat{\F}}$ such that ${\widehat{\E}}(\int_{0}^{T}|\widehat{Z}_{s}(.)|^{2}ds)<+\infty$ and
$$M(.)=M_{0}(.)+\int_{0}^{T}\widehat{Z}_{s}(.)dB_{s}, \quad \widehat{\p}-a.e,$$ where $M_{0}(.)\in{\widehat{\f}}_{0}.$
 \end{thm}

 \begin{proof}
First we suppose that $M(.)\in BL^{2}$.
   Since $SBL^{2}$ is dense in $BL^{2}$, for each $M(.)\in BL^{2}$ there exists a sequence $\{M^{n}(.)\}_{n\in\N}\in SBL^{2}$ such that $M^{n}(.)\rightarrow M(.)$ in $L^{2}$. Thus, by linearity it suffices to prove the theorem for $M(\omega,u)=H(\omega)K(u)\in SBL^{2}.$ Since $H\in L^{2}({\Omega},{\f}_{T},{\p})$ and $\F$ is a Brownian filtration, from the martingale representation theorem, there exists a unique ${\f}-$measurable process $Z=(Z_t)_{t\in[0,T]}$ which is predictable w.r.t $\F$
   and ${\E}(\int_{0}^{T}|Z_s|^{2}ds)<+\infty$ such that $$H=H_{0}+\int_{0}^{T}Z_sdB_{s}, \quad {\p}-a.s,$$
   where $H_{0}\in{\f}_{0}$.
 By multiplying by $K,$ we get for $ \eta-a.e\ u\in\R$
\begin{eqnarray}
\quad M(u)=M_{0}(u)+\int_{0}^{T}{{\widehat{Z}}_{s}(u)}dB_{s}, \quad {\p}-a.s.
 \label{rabete}\end{eqnarray}
    Where $M_{0}(u):=H_{0}K(u)$ and
   ${\widehat{Z}}_{s}(u):={Z}_{s}K(u).$\\
   It can be easily seen that $M_{0}(.)\in{\widehat{\f}}_{0}$ and ${\widehat{Z}}(.)$ is $\widehat{\F}-$predictable. Furthermore, from boundedness of $K$, we have
   $${\widehat{\E}}(\int_{0}^{T}|{\widehat{Z}}_{s}(.)|^{2}ds)={\E}({\int_{0}^{T}|{{Z}}_{s}|^{2}ds})(\int_{\R}K^{2}(u)d\eta(u))<+\infty.$$
   Since the null sets are independent of $u,$ we have
   $$M(.)=M_{0}(.)+\int_{0}^{T}{{\widehat{Z}}_{s}(.)}dB_{s}, \quad \widehat{\p}-a.e.$$


   Now for $M(.) \in BL^{2},\ M^{n}(.)\rightarrow M(.)$ in $L^{2}$ where $\{M^{n}(.)\}_{n\in\N}\in SBL^{2}$. Thus, $\{M^{n}(.)\}_{n\in\N}$ is Cauchy in $L^{2}$.
   On the other hand, we have

   \begin{eqnarray*}
  {\widehat{\E}}\left(|M^{n}(.)-M^{m}(.)|^{2}\right)&=&{{\widehat{\E}}\left(|M^{n}_{0}(.)-M^{m}_{0}(.)|^{2}\right)+{\widehat{\E}}\left(|\int_{0}^{T}({\widehat{Z}}^{n}_{s}(.)-{\widehat{Z}}^{m}_{s}(.))dB_{s}|^{2}\right)}\\
    &+& 2{{\widehat{\E}}\left(|M^{n}_{0}(.)-M^{m}_{0}(.)|\left(|\int_{0}^{T}({\widehat{Z}}^{n}_{s}(.)-{\widehat{Z}}^{m}_{s}(.))dB_{s}|\right)\right)}.
 \end{eqnarray*}

   From the boundedness of $M^{n}_{0}(.)$ for $n\geq1,$ and Proposition \ref{martingale} we get
    \begin{eqnarray*}
  {\widehat{\E}}\left(|M^{n}(.)-M^{m}(.)|^{2}\right)={{\widehat{\E}}\left(|M^{n}_{0}(.)-M^{m}_{0}(.)|^{2}\right)+{\widehat{\E}}\left(|\int_{0}^{T}({\widehat{Z}}^{n}_{s}(.)-{\widehat{Z}}^{m}_{s}(.))dB_{s}|^{2}\right)},
 \end{eqnarray*}
since
 \begin{eqnarray*}
  {\widehat{\E}}\left(|M^{n}_{0}(.)-M^{m}_{0}(.)|\left(|\int_{0}^{T}({\widehat{Z}}^{n}_{s}(.)-{\widehat{Z}}^{m}_{s}(.))dB_{s}|\right)\right)&\leq& c{\widehat{\E}}\left(\int_{0}^{T}|{\widehat{Z}}^{n}_{s}(.)-{\widehat{Z}}^{m}_{s}(.)|dB_{s}\right)
 \end{eqnarray*}
and $\int_{0}^{T}|{\widehat{Z}}^{n}_{s}(.)-{\widehat{Z}}^{m}_{s}(.)|dB_{s}$ is a martingale w.r.t $\widehat{\F}$ with zero expectation.
Therefore $\{M^{n}_{0}(.)\}_{n\in\N}$  is Cauchy in $L^{2}$, and from Corollary \ref{cor} $\{{\widehat{Z}}^{n}(.)\}_{n\in\N}$ is Cauchy in $L^{2}(\widehat{\Omega}\times[0,T])$. Thus the sequences converge to $M_{0}(.)$ resp. $\widehat{Z}(.)$.  A subsequence of $\{M^{n}_{0}(.)\}$ converges to $M_{0}(.)$ for $\widehat{\p}-a.e.\, (\omega,u)\in\widehat{\Omega}$. Therefore $M_{0}(.)$ is $\widehat{\f}_{0}-$measurable. Similarly, by extracting a subsequence we obtain that $\widehat{Z}(.)$ is $\widehat{\F}-$predictable after an eventual modification on a set of measure zero in product space.
By using Corollary \ref{cor}  and uniqueness of limits in $L^{2}$, the proof of existence of a representation is complete for $M(.) \in BL^{2}$.
\quad

Finally for $M(.)\in L^{2},$ we define $M^{n}(.):=M(.)\cdot 1_{\{|M(.)|\le n\}},\ n\in\N$. Then $\{M^{n}(.)\}_{n\in\N}$ is a sequence of bounded random variables. Since $\widehat{\E}(|M(.)|^{2})<\infty,$ we have $M^{n}(.)\rightarrow M(.)$ in $L^2$. On the other hand, since $M^{n}(.)\in BL^{2},\ n\in\N,$  for each individual $n$ we get a representation of the form
$$M^{n}(.)=M^{n}_{0}(.)+\int_{0}^{T}\widehat{Z}^{n}_{s}(.)dB_{s},\quad\widehat{\p}-a.e.$$
where $M^{n}_{0}(.)\in\widehat{\f}_{0}$ and $\widehat{Z}^{n}(.)$ is an $\widehat{\F}-$predictable process. Using Corollary \ref{cor} again, in a similar way as in the preceding part of the proof we obtain the assertion for $M(.)\in L^2.$
\quad

To prove uniqueness, suppose that there are two predictable processes ${\widehat{Z}}^{1}(.)$ and ${\widehat{Z}}^{2}(.)$ such that
$$M(.)=M_{0}(.)+\int_{0}^{T}\widehat{Z}^{1}_{s}(.)dB_{s}=M_{0}(.)+\int_{0}^{T}\widehat{Z}^{2}_{s}(.)dB_{s}, \quad \widehat{\p}-a.e.$$

Then from Corollary \ref{cor} we get
$$0=\widehat{\E}\left(\int_{0}^{T}({\widehat{Z}}^{1}_{s}(.)-{\widehat{Z}}^{2}_{s}(.))dB_{s}\right)^{2}=\widehat{\E}\left(\int_{0}^{T}({\widehat{Z}}^{1}_{s}(.)-{\widehat{Z}}^{2}_{s}(.))^{2}d{s}\right).$$
This implies that for $a.a\,(\omega,u,s)\in\widehat{\Omega}\times[0,T]$ we have\ ${\widehat{Z}}^{1}_{s}(\omega,u)={\widehat{Z}}^{2}_{s}(\omega,u).$
 \end{proof}
\begin{thm}
 Let $M:\widehat{\Omega}\times[0,T]\rightarrow\R$ be an $\widehat{\f}\otimes\B([0,T])-$measurable function
 such that $\{M_{t}(.)\}_{t\in[0,T]}$ is a martingale w.r.t  ${\widehat{\F}}$ and ${\widehat{\E}}(|M_{t}(.)|^{2})<+\infty,$ for $t\in[0,T]$. Then there exists a unique $\widehat{\f}-$measurable
 function $\widehat{Z}:\widehat{\Omega}\times[0,T]\rightarrow\R,$
 which is predictable w.r.t ${\widehat{\F}}$ such that ${\widehat{\E}}(\int_{0}^{T}|\widehat{Z}_{s}(.)|^{2}ds)<+\infty$ and for $t\in[0,T],$
$$M_{t}(.)=M_{0}(.)+\int_{0}^{t}\widehat{Z}_{s}(.)dB_{s}, \quad \widehat{\p}-a.e.$$
where $M_{0}(.)$ is ${\widehat{\f}}_{0}-$measurable.
\label{MRT}\end{thm}
\begin{proof}
From the martingale property, since $M_{T}(.)\in L^{2}$ we have for $t\in[0,T]$
$$M_{t}(.)=\widehat{\E}(M_{T}(.)|\widehat{\f}_{t}), \quad \widehat{\p}-a.e.$$
Thus from the previous theorem, there exists a unique $\widehat{\F}_0-$measurable $M_0(.)\in L^2$ and a unique $\widehat{\F}-$predictable process ${\widehat{Z}}(.)$ such that $$M_{T}(.)=M_{0}(.)+\int_{0}^{T}{\widehat{Z}}_{s}(.)dB_{s}, \quad \widehat{\p}-a.e.$$ Therefore for $t\in[0,T]$
 $$M_{t}(.)=\widehat{\E}(M_{T}(.)|\widehat{\f}_{t})=M_{0}(.)+\int_{0}^{t}{\widehat{Z}}_{s}(.)dB_{s}, \quad \widehat{\p}-a.e.$$
\end{proof}
Since our study will be based on the connection between RBSDE and optimal stopping problems, we shall restrict our attention to the case in which a generator $f$ of an RBSDE is just an $\widehat{\F}-$progressively measurable process.
Now by employing the representation property for martingales depending on a parameter of the preceding theorem, we can define and solve parametrized reflected BSDE in the product space. For a probability measure $Q$ on $(\widehat{\Omega},{\widehat{\f}},{\widehat{\F}})$, we consider the following spaces

\begin{eqnarray*}
\widehat{\mathcal{L}}^2_{Q}&=&\{X(.): X(.) \,\,\widehat{\f}_{T}-\mbox{measurable random variable}, \,\,\widehat{\E}^{Q}(|X(.)|^{2}) < \infty\},\\
\widehat{\mathcal{H}}^2_{Q}&=&\{X(.): X(.) = (X_{t}(.))_{0\le t\le T}\,\, \mbox{continuous predictable process,}\,\,\widehat{\E}^{Q}\displaystyle\int_0^T|X_{t}(.)|^2 dt < \infty\},\\
\widehat{\mathcal{S}}^2_{Q}&=&\{X(.): X(.) = (X_{t}(.))_{0\le t\le T}\,\,\mbox{continuous predictable process,}\,\,\widehat{\E}^{Q}(\displaystyle{\sup_{0\le t\le T}} |X_{t}(.)|^2)<\infty\},\\
\widehat{\mathcal{I}}_{Q}&=&\{K(.): K(.) = (K_{t}(.))_{0\le t\le T}\,\,\mbox{non-decreasing continuous process,}\,\, K_{0}(.)=0,\,\,K_{T}(.)\in\widehat{\mathcal{L}}^2_{Q}\}.
\end{eqnarray*}

$\widehat{\E}^Q$ stands for the expectation w.r.t the measure $Q$. If $Q=\widehat{\p}$, we simply write $\widehat{\E}$ and also $\widehat{\mathcal{S}}^2, \widehat{\mathcal{H}}^2, \widehat{\mathcal{I}}$.

\quad

Now consider the filtered probability space $(\widehat{\Omega},{\widehat{\f}},{\widehat{\F}},\widehat{\p})$. From Corolary \ref{mre}, $B$ is still a Brownian motion w.r.t this probability space.
Now let a triplet $(\xi(.), f(.), L(.))$ be given satisfying

\label{shart2}
\quad(i') $\xi(.)\in\widehat{\mathcal{L}}^2;$

\quad(ii') $f(.)$ is a predictable process s.t. $\widehat{\E}{\left[\int_{0}^{T}f^{2}_{t}(.)dt\right]}<\infty;$

\quad(iii') $L(.)\in\widehat{\mathcal{S}}^2$.

We call a triplet $(Y(.),Z(.),K(.))\in\widehat{\mathcal{S}}^2\times\widehat{\mathcal{H}}^2\times\widehat{\mathcal{I}}$ a solution of the parametrized RBSDE with driver $f(\cdot),$ terminal variable $\xi(\cdot)$, and barrier $L(\cdot)$, if

\begin{equation}
\left\{
\begin{aligned}
Y_{t}(.)&=\xi(.)+\int_{t}^{T}{f_{s}(.)ds}+K_{T}(.)-K_{t}(.)-{\int_{t}^{T}{Z_{s}(.)dB_{s}}}\quad\quad 0\leq t \leq T,\\
Y_{t}(.)&\geq L_{t}(.),\quad\quad0\leq t \leq T,\ \quad\quad\ \int_{0}^{T}{(Y_{t}(.)-L_{t}(.))dK_{t}(.)}=0.\\
\end{aligned}
\right.
\label{parametrized 2}\end{equation}






  \begin{rem}
For $\eta$ a.e $u\in\R,$
\begin{equation}
\left\{
\begin{aligned}
Y_{t}(u)&=\xi(u)+\int_{t}^{T}{f_{s}(u)ds}+K_{T}(u)-K_{t}(u)-{\int_{t}^{T}{Z_{s}(u)dB_{s}}}\quad\quad 0\leq t \leq T,\\
Y_{t}(u)&\geq L_{t}(u),\quad\quad0\leq t \leq T,\ \quad\quad\ \int_{0}^{T}{(Y_{t}(u)-L_{t}(u))dK_{t}(u)}=0.
\end{aligned}
\right.
\end{equation}

is an RBSDE w.r.t $\left(\Omega,\f,\F,\p\right).$ This is the reason we call the RBSDE (\ref{parametrized 2}) a parametrized RBSDE.
 \end{rem}

Similarly to the usual case, we shall always assume that $L_T(.)\leq \xi(.)$. Equipped with these concepts, we can extend the classical existence and uniqueness theorem for solutions of RBSDE to the parametrized RBSDE (\ref{parametrized 2}), and then rewrite Proposition \ref{optimal} for the product space in the following remark.

\begin{thm}
Under assumptions  (i'), (ii'), and (iii'), the RBSDE (\ref{parametrized 2}) has a unique solution $(Y(.),Z(.),K(.))$.
\label{mthm}\end{thm}

\begin{proof}
From (ii'), $f$ is an $\widehat{\F}-$progressively measurable process such that $\widehat{\E}{\left[\int_{0}^{T}f^{2}_{t}(.)dt\right]}<\infty.$ Thus the RBSDE (\ref{parametrized 2}) is a backward reflection problem (BRP) according to the terminology of \citep{Nichole}. Now we can rewrite the proof of Proposition 5.1 in \citep{Nichole} for our BRP on the product space:
Most of the proof is similar. Hence we only mention the main steps.
To prove the existence of the solution, we introduce the process $Y(.)=(Y_{t}(.))_{t\in[0,T]}$ defined by
$$Y_{t}(.)=\esssup_{{\tau(.)\in\mathcal{T}_{t,T}\left(\widehat{\F}\right)}}\widehat{\E}{\left[\int_{t}^{\tau(.)}{f_{s}(.) ds+L_{\tau(.)}(.)1_{[0,T[}(\tau(.))+\xi(.)1_{\{T\}}(\tau(.))}|\widehat{\f}_{t}\right]},\ t\in[0,T].$$

Then with the argument given in \citep{Nichole}, $Y_{t}(.)+\int_{0}^{t}f_{s}(.) ds$ is the value function of an optimal stopping problem with payoff
$$H_{t}(.)=\int_{0}^{t}{f_{s}(.) ds}+L_{t}(.)1_{[0,T[}(t)+\xi(.)1_{\{T\}}(t).$$ By the theory of Snell envelopes, it is the smallest supermartingale which dominates $H(.)$. $Y(.)$ is continuous because of the continuity of $H(.)$ on the interval $[0,T)$ and the assumption $L_{T}(.)\leq\xi(.).$ This means that the jump of $H(.)$ at time $T$ is positive. So $Y(.)\in\widehat{\mathcal{S}}^2$ from the following inequality
$$\widehat{\E}\left(\sup_{0\leq t\leq T}Y^{2}_{t}(.)\right)\leq c\widehat{\E}\left(\xi^{2}(.)+\int_{0}^{T}f^{2}_{s}(.) ds+\sup_{0\leq t\leq T}L^{2}_{t}(.)\right).$$
which is obtained by Burkholder's inequality and the conditions (i'), (ii'), and (iii').
Denote by $\tau^{*}(.)$ the stopping time $$\tau^{*}(.)=\inf\{t\leq s\leq T : Y_{s}(.)\leq L_{s}(.)\}\wedge T.$$
Then $\tau^{*}(.)$ is optimal, in the sense that
\begin{equation}
Y_{t}(.)=\widehat{\E}\left[\int_{t}^{\tau^{*}(.)}f_{s}(.)ds+L_{\tau^{*}(.)}(.)1_{[0,T[}(\tau^{*}(.))+\xi(.)1_{\{T\}}(\tau^{*}(.))|\widehat{\f}_{t}\right]
\end{equation}

Now Doob-Meyer's decomposition of the continuous supermartingale $Y_{t}(.)+\int_{0}^{t}f_{s}(.) ds$ yields an adapted continuous process $K(.)=(K_{t}(.))_{t\in[0,T]}$ and a continuous uniformly integrable martingale $M(.)=(M_{t}(.))_{t\in[0,T]}$ such that $$Y_{t}(.)=M_{t}(.)-\int_{0}^{t}f_{s}(.) ds-K_{t}(.),$$ where $K_{0}(.)=0$ and $K_{t}=K_{\tau^{*}(.)}$. The Skorohod condition and square integrability of $K_{T}(.)$ follow from arguments similar to the ones of \citep{Nichole} in the product space. Hence the $\widehat{\F}-$martingale $$M_{t}(.)=\widehat{\E}(M_{T}(.)|\widehat{\f}_{t})=\widehat{\E}\left(\xi(.)+\int_{0}^{T}f_{s}(.)ds-K_{T}(.)|\widehat{\f}_{t}\right)$$ is also square integrable, i.e. $\widehat{\E}(|M_{t}(.)|^{2})<\infty,\ t\in[0,T].$
Thus, we can use Theorem \ref{MRT} to find the process $\widehat{Z}(.)$ such that $M_{t}(.)=\int_{0}^{t}\widehat{Z}_{s}(.)ds,$ where ${\widehat{\E}}(\int_{0}^{T}|\widehat{Z}_{s}(.)|^{2}ds)<+\infty$.

Uniqueness of the solution can be achieved from Corollary 3.7 in \citep{Nichole} which is satisfied on the product space under assumptions (i'), (ii'), and (iii').

\end{proof}

\begin{rem}
 Under the assumptions (i'), (ii'), and (iii'),
 $Y_{t}(.)+\int_{0}^{t}{f_{s}(.)ds}$ is the value function of an optimal stopping problem with the payoff
 $$\int_{0}^{t}f_{s}(.)ds+L_{t}(.)1_{[0,T[}(t)+\xi(.)1_{\{T\}}(t),$$
 where $Y_t(.)$ is the solution of the RBSDE (\ref{parametrized 2}). Furthermore the stopping time
$\tau^{*}(.)=\inf{\{s\in[t,T] : Y_s(.)=L_s(.)\}}\wedge{T}$ is optimal, in the sense $$Y_{t}(.)=\widehat{\E}{\left[\int_{t}^{\tau^{*}(.)}{f_{s}(.)ds+L_{\tau^{*}(.)}(.)1_{[0,T[}(\tau^{*}(.))+\xi(.)1_{\{T\}}(\tau^{*}(.))}|\widehat{\f}_{t}\right]}.$$

Especially in the case $f\equiv0$, $Y_t(.)$, the solution of the RBSDE (\ref{parametrized 2}),
is the value function of an American contingent claim with the payoff $L_{t}(.)1_{[0,T[}(t)+\xi(.)1_{\{T\}}(t)$ and $\tau^{*}(.)$ is
the optimal stopping time for the buyer.
 \label{rem2} \end{rem}

\begin{rem}
An extension of Theorem \ref{mthm} to the case in which $f$ also depends on $y$ and $z$ with global Lipschitz continuity in these two variables can be obtained by means of the proof of Theorem 5.2 in \citep{Nichole}.
\end{rem}
 \subsection{RBSDE in an initially enlarged filtration}

We will now show that under suitable conditions on the parametrized payoff function $R$ in (\ref{newpay}), the corresponding value function is the solution of
a parametrized RBSDE on the same product space. For this purpose, consider the product space $(\widehat{\Omega},{\widehat{\f}},{\widehat{\F}},\widehat{\p})$ from (\ref{product}) where $\widehat{\p}=\p\otimes P^{G}$ and $P^{G}$ is the law of the random variable $G$
which carries the extra information.
We consider the RBSDE (\ref{parametrized 2}) with $f\equiv0$, $L(.)=L\alpha(.)$, and $\xi(.)=\xi\alpha_{T}(.),$ where
$L$ and $\xi$ are the barrier resp. final variable of the usual RBSDE (\ref{BSDE}). Then we obtain the following parametrized RBSDE

\begin{equation}
\left\{
\begin{aligned}
-dY_{t}(.)&=dK_{t}(.)-Z_{t}(.)dB^{\F}_{t},\quad\quad 0\leq t \leq T,\\
Y_{T}(.)&=\xi\alpha_{T}(.),\\
Y_{t}(.)&\geq L_{t}\alpha_{t}(.),\quad 0\leq t \leq T,\quad\quad\textstyle\int_{0}^{T}{(Y_{t}(.)-L_{t}\alpha_{t}(.))dK_{t}(.)}=0.\\
\end{aligned}
\right.
\label{dparametrized}\end{equation}

Since we work with two different filtrations in this section, we denote by $B^{\F}_{t}$ a Brownian motion w.r.t $\F$. From Theorem \ref{mthm} and Remark \ref{rem2} in the previous section, under conditions (i') and (iii') for $\xi\alpha_{T}(.)$ and $L_{t}\alpha_{t}(.)$, the RBSDES (\ref{dparametrized}) has a unique solution $(Y(.),Z(.),K(.))\in\widehat{\mathcal{S}}^2\times\widehat{\mathcal{H}}^2\times\widehat{\mathcal{I}}$ and
$Y_{t}(.)$ is the value function of an optimal stopping problem with the payoff $R_t(\cdot)=L_{t}\alpha_{t}(.)1_{[0,T[}( t)+\xi\alpha_{T}(.)1_{\{T\}}(t), t\in[0,T].$

\quad

Theorem \ref{thm1} motivates us to define $\widehat{Y}(.):=\frac{Y(.)}{\alpha(.)} $. Recall that, due to our hypotheses on $G$, $\alpha(.)$ is a positive continuous martingale, so that $\sup_{s\in[0,T]}{\frac{1}{\alpha^2_{s}(.)}}<\infty.$ This implies that our definition makes sense. We will prove that $\widehat{Y}(G)$ is the solution of an RBSDE that corresponds to the optimization problem in the enlarged filtration.
Note that for each $u\in \R,$ $\alpha(u)$ is a martingale w.r.t $\F$ and for each $t\in[0,T],$ it has an $\widehat{\f}-$measurable version.
Therefore from Proposition \ref{martingale}, $\{\alpha_{t}(.)\}_{t\in[0,T]}$ is a martingale w.r.t $\widehat{\F}$. If we suppose that it is $\widehat{\p}-$square integrable, then the martingale representation Theorem
\ref{MRT} yields $d\alpha_{t}(.)=\beta_{t}(.)dB^{\F}_{t}$, where $\beta(.)$ is an $\widehat{\F}-$predictable process which is square integrable with respect to $\widehat{\p}$.\\
 By Ito's formula, we get that $\widehat{Y}(.)$ satisfies in the following RBSDE:

\begin{equation}
\left\{
\begin{aligned}
-d{\widehat{Y}}_{t}(.)&=-\left[(\frac{\beta_{t}(.)}{\alpha_{t}(.)})^{2}\widehat{Y}_{t}(.)-{\frac{Z_{t}(.)}{\alpha^{2}_{t}(.)}}\beta_{t}(.)\right]dt+\frac{1}{\alpha_{t}(.)}dK_{t}(.)-\left[\frac{Z_{t}(.)}{\alpha_{t}(.)}-\frac{\beta_{t}(.)}{\alpha_{t}(.)}\widehat{Y}_{t}(.)\right]dB^{\F}_{t},\ \ 0\leq t\leq T,\\
\widehat{Y}_{T}(.)&=\xi,\\
\widehat{Y}_{t}(.)&\geq L_{t},\quad 0\leq t\leq T,\quad\quad\textstyle\int_{0}^{T}{(\widehat{Y}_{t}(.)-L_{t})\alpha_{t}(.)dK_{t}(.)}=0.
\end{aligned}
\right.
\label{idparametrized}\end{equation}

The Skorokhod condition has the stated form because
\begin{eqnarray*}
 \int_{0}^{T}{(\widehat{Y}_{t}(.)-L_{t})\alpha_{t}(.)d{K}_{t}(.)}&=&\int_{0}^{T}{(\frac{Y_{t}(.)}{\alpha_{t}(.)}-L_{t})\alpha_{t}(.)dK_{t}(.)}\\
 &=&\int_{0}^{T}{(Y_{t}(.)-L_{t}\alpha_{t}(.))dK_{t}(.)}=0, \quad\widehat{\p}-a.e.
\end{eqnarray*}

We now define $\widehat{K}(.)=\int_{0}^\cdot{\frac{1}{\alpha_{s}(.)}dK_{s}(.)}$ and $\widehat{Z}(.)=\frac{Z(.)}{\alpha(.)}-{\frac{\beta(.)}{\alpha(.)}}\widehat{Y}(.).$
Since $\alpha(.)$ is continuous in $t$ and positive, $\widehat{K}(.)$ is an increasing continuous process such that $\widehat{K}_{0}\equiv0$ and $d\widehat{K}_t(.)=\frac{dK_{t}(.)}{\alpha_{t}(.)}.$
Furthermore, $\widehat{Y}(.), \widehat{K}(.)$ and $\widehat{Z}(.)$ are $\widehat{\F}-$predictable processes. This follows from the $\widehat{\F}-$predictability of $Y(.), Z(.), K(.), \alpha(.)$, and $\beta(.)$. In addition, we have

\begin{eqnarray*}
 \int_{0}^{T}{\widehat{Z}_s}^{2}(.)ds &\leq & 2 \int_{0}^{T}\left(\frac{Z_s(.)}{\alpha_{s}(.)}\right)^{2}ds+2\int_{0}^{T}\left(\frac{\beta_s(.)}{\alpha_{s}(.)}\widehat{Y}_{s}(.)\right)^{2}ds\\
 &\leq &2\sup_{s\in[0,T]}{\frac{1}{\alpha^2_{s}(.)}}\int_{0}^{T} {Z^{2}_{s}(.)}ds+2\left(\sup_{s\in[0,T]}Y^{2}_{s}(.)\right)\left(\sup_{s\in[0,T]}{\frac{1}{\alpha^{4}_{s}(.)}}\right)\int_{0}^{T}{\beta^{2}_{s}(.)ds}<\infty, \quad \widehat{\p}-a.e.,
\end{eqnarray*}

because  $Y(.)$ is continuous 
 in $t$, $\alpha(.)$  continuous and strictly positive, and $Z(.)$ and $\beta(.)$ are square integrable in $\widehat\Omega\times[0,T]$. Thus the Ito integral process for $\widehat{Z}$ with respect to $B^{\F}$ is still defined and is a local martingale (see \citep{Oks}, p. 35). With similar arguments, it can be shown that

\begin{eqnarray*}
 \sup_{s\in[0,T]}|\widehat{Y}_s(.)|^{2} \leq \left(\sup_{s\in[0,T]}{\frac{1}{\alpha^2_{s}(.)}}\right)\left(\sup_{s\in[0,T]}|Y_{s}(.)|^{2}\right)<\infty, \quad \widehat{\p}-a.e.
\end{eqnarray*}

Furthermore, since $K(.)\in\widehat{\mathcal{I}}, $ we have

\begin{eqnarray*}
 {\widehat{K}}^{2}_{T}(.)\leq \left(\sup_{s\in[0,T]}{\frac{1}{\alpha^2_{s}(.)}}\right) K_{T}^{2}(.)<\infty, \quad \widehat{\p}-a.e.
\end{eqnarray*}

Now we introduce the following spaces, corresponding to a filtration $\mathbb H=(\mathcal{H}_{t})_{t\in[0,T]}$ on an arbitrary probability space:
\begin{eqnarray*}
\bar{\mathcal{H}}^2_{\mathbb H}&=& \{X: X = (X_{t})_{0\le t\le T}\,\, \mathbb H-\mbox{predictable process}, \displaystyle\int_0^T|X_{t}|^2 dt < \infty\},\\
\bar{\mathcal{S}}^2_{\mathbb H}&=&\{X: X= (X_{t}(.))_{0\le t\le T}\,\,\mbox{continuous}\,\,\mathbb H-\mbox{predictable process}, \displaystyle{\sup_{0\le t\le T}} |X_{t}|^2<\infty\},\\
\bar{\mathcal{I}}_{\mathbb H}&=&\{K: K= (K_{t})_{0\le t\le T}\,\,\mbox{increasing continuous process,}\,\, K_{0}=0,\,\, K_{T}\,\,\mathcal H_{T}-\mbox{measurable},\,\, K^{2}_{T}<\infty.\}
\end{eqnarray*}

Therefore $(\widehat{Y}_{t}(.),\widehat{Z}_{t}(.),\widehat{K}_{t}(.))_{0 \leq t \leq T}\in\bar{\mathcal{S}}^2_{\widehat{\F}}\times\bar{\mathcal{H}}^2_{\widehat{\F}}\times\bar{\mathcal{I}}_{\widehat{\F}}$ solves the RBSDE
\begin{equation}
\left\{
\begin{aligned}
-d{\widehat{Y}}_{t}(.)&=\frac{\beta_{t}(.)}{\alpha_{t}(.)}\widehat{Z}_{t}(.)dt+d{\widehat{K}}_{t}(.)-\widehat{Z}_{t}(.)dB^{\F}_{t},\quad\quad 0\leq t \leq T,\\
\widehat{Y}_{T}(.)&=\xi,\\
\widehat{Y}_{t}(.)&\geq L_{t},\quad 0\leq t \leq T,\quad\quad\textstyle\int_{0}^{T}{(\widehat{Y}_{t}(.)-L_{t})d\widehat{K}_{t}(.)}=0,
\end{aligned}
\right.
\label{newrbsde}\end{equation}

in $(\widehat{\Omega},{\widehat{\f}},{\widehat{\F}},\widehat{\p})$.

\quad

The Skorokhod condition follows from (\ref{idparametrized}), since
\begin{eqnarray*}
 \int_{0}^{T}{(\widehat{Y}_{t}(.)-L_{t})d\widehat{K}_{t}(.)}&=&\int_{0}^{T}{(\widehat{Y}_{t}(.)-L_{t})\frac{1}{\alpha^{2}_{t}(.)}\alpha_{t}(.)d{K}_{t}(.)}\\
 &\leq& (\sup_{t\in[0,T]}{\frac{1}{\alpha^{2}_{t}(.)}})\textstyle\int_{0}^{T}{(\widehat{Y}_{t}(.)-L_{t})\alpha_{t}(.)d{K}_{t}(.)}=0.
\end{eqnarray*}

The last inequality holds by $\widehat{Y}(.)\geq L$, and since $\alpha(.)$ is positive and continuous.
The following proposition recalls the canonical decomposition of a local martingale in the smaller filtration with respect to the larger one.
\begin{pro}
 Any $\F-$local martingale $M$ is a $\G-$semimartingale with canonical decomposition
 $$M_{t}=M^{G}_{t}+\int_{0}^{t}{\frac{d<M,\alpha_{.}(G)>_{s}}{\alpha_{s^{-}}(G)}},$$
 where $M^{G}$ is a $\G-$local martingale.
\end{pro}

\begin{proof}
 See Theorem 2.5.c in \citep{jacod2}. Also \citep{Amendinger} and \citep{Zargar}.
\end{proof}

The preceding proposition and the continuity of $\alpha(.)$ imply

\begin{equation}
 B^{\F}_{t}=B^{\G}_{t}+\int_{0}^{t}{\frac{d<B^{\F},\alpha_{.}(G)>_{s}}{\alpha_{s^{-}}(G)}}=B^{\G}_{t}+\int_{0}^{t}{\frac{\beta_{s}(G)}{\alpha_{s^{}}(G)}}ds.
\label{brownian}\end{equation}

Now consider $(\widehat{Y}_{t}(G),\widehat{Z}_{t}(G),\widehat{K}_{t}(G))_{0 \leq t \leq T}$. By Remark \ref{remprog}, it is a triplet of $\G-$predictable processes. Evaluating (\ref{newrbsde}) at $G$ and replacing $B^{\F}_{t}$ from the above proposition, $(\widehat{Y}(G),\widehat{Z}(G),\widehat{K}(G))\in\bar{\mathcal{S}}^2_{\G}\times\bar{\mathcal{H}}^2_{\G}\times\bar{\mathcal{I}}_{\G}$  will satisfy in the following RBSDE in $(\Omega,\f,\G,\p)$:

\begin{equation}
\left\{
\begin{aligned}
-d{\widehat{Y}}_{t}(G)&=d{\widehat{K}}_{t}(G)-\widehat{Z}_{t}(G)dB^{\G}_{t},\quad\quad 0\leq t \leq T,\\
\widehat{Y}_{T}(G)&=\xi,\\
\widehat{Y}_{t}(G)&\geq L_{t},\quad 0\leq t \leq T,\quad\quad\textstyle\int_{0}^{T}{(\widehat{Y}_{t}(G)-L_{t})d\widehat{K}_{t}(G)}=0.
\end{aligned}
\right.
\label{grbsde}\end{equation}

RBSDE (\ref{grbsde}) is an RBSDE in the initially enlarged filtration $\mathbb{G}$ with generator $f\equiv 0$. As we will see in the following section, it relates to our optimal stopping problem in the initially enlarged filtration. We will comment on the square-integrability of the solution components below.

RBSDE (\ref{newrbsde}) possesses a non-trivial driver independent of $y$. 
 Similarly to SDE we can apply Girsanov's theorem to get rid of it. To do this,
we set for $t\in[0,T]$ $q_{t}(.):=exp\left(\int_{0}^{t}{\frac{\beta_{s}(.)}{\alpha_{s}(.)}}dB^{\F}_{s}-\frac{1}{2}\int_{0}^{t}{(\frac{\beta_{s}(.)}{\alpha_{s}(.)})^{2}ds}\right).$
Then Girsanov's theorem implies that if $\frac{\beta(.)}{\alpha(.)}$ satisfies Novikov's condition which means
\begin{equation}\label{novikov}
\widehat{\E}\left(exp(\int_{0}^{T}{\frac{1}{2}(\frac{\beta_{s}(.)}{\alpha_{s}(.)})^{2}ds})\right)<\infty,
\end{equation}
then $q_{T}(.)$ is a likelihood ratio which defines a new probability measure on $(\widehat{\Omega},\widehat{\f})$ by $\widehat{Q}(A)=\widehat{\E}\left(q_{T}(.)1_{\{A\}}(.)\right),\ A\in\widehat{\f}$ , under which $\widehat{B}_{t}(.):=B^{\F}_{t}-\int_{0}^{t}{\frac{\beta_{s}(.)}{\alpha_{s}(.)}}ds$
is a Brownian motion. We now suppose that (\ref{novikov}) is satisfied. Under the probability measure $\widehat{Q}$ on the space $(\widehat{\Omega},\widehat{\f},\widehat{\F})$ we rewrite (\ref{newrbsde}) to get the following RBSDE with standard parameters $(\xi,0,L)$ w.r.t the Brownian motion $\widehat{B}(.),$
\begin{equation}
\left\{
\begin{aligned}
-d{\widehat{Y}}_{t}(.)&=d{\widehat{K}}_{t}(.)-\widehat{Z}_{t}(.)d{\widehat{B}}_{t}(.),\quad\quad 0\leq t \leq T,\\
\widehat{Y}_{T}(.)&=\xi,\\
\widehat{Y}_{t}(.)&\geq L_{t},\quad 0\leq t \leq T,\quad\quad\textstyle\int_{0}^{T}{(\widehat{Y}_{t}(.)-L_{t})d\widehat{K}_{t}(.)}=0.
\end{aligned}
\right.
\label{Qrbsde}\end{equation}

Note that $\widehat{B}(G)=B^{\G}$ from (\ref{brownian}).
We know from \citep{Nichole} that if
$$\mbox{(i*)} \,\,\xi\in\widehat{\mathcal{L}}^2_{\widehat{Q}},$$ and
$$\mbox{(iii*)}\,\, L\in\widehat{\mathcal{S}}^2_{\widehat{Q}},$$
then (\ref{Qrbsde}) has a unique solution $(\widehat{Y}(.),\widehat{Z}(.),\widehat{K}(.))\in\widehat{\mathcal{S}}^2_{\widehat{Q}}\times\widehat{\mathcal{H}}^2_{\widehat{Q}}\times\widehat{\mathcal{I}}_{\widehat{Q}}$.
Moreover, since $\alpha(.)$ is strictly positive, we have
 \begin{eqnarray*}
  d\alpha_{t}(.)=\beta_{t}(.)dB^{\F}_{t}=\frac{\beta_{t}(.)}{\alpha_{t}(.)}\alpha_{t}(.)dB^{\F}_{t}.
 \end{eqnarray*}

Thus, Ito's formula gives

$$\alpha(.)=exp\left(\int_{0}^\cdot{\frac{\beta_{s}(.)}{\alpha_{s}(.)}}dB^{\F}_{s}-
\frac{1}{2}\int_{0}^\cdot{(\frac{\beta_{s}(.)}{\alpha_{s}(.)})^{2}ds}\right)=q(.),$$

and $\alpha(.)$ acts as a likelihood ratio between $\widehat{\p}$ and $\widehat{Q}$.

\begin{rem}
From the definition of $\widehat{Q}$, it can be easily seen that the assumptions (i*) and (iii*) are equivalent with (i') and (iii') for $\xi\alpha(.)$ and $L\alpha(.)$, if $\alpha(.)$ is bounded $\widehat{\p}-$a.e.
\label{frem}\end{rem}

Therefore we may state that an initial enlargement of a filtration in optimal stopping problems corresponds to
a change of a measure in a parametrized RBSDE on the product of the underlying probability space and the state space in which the additional information $G$ takes its values. See \citep{Nick} for a complete discussion. Novikov's condition is satisfied for example if
$\frac{\beta_{t}(.)}{\alpha_{t}(.)}$ is $\widehat{\p}-$a.e. bounded. This condition has been studied in \citep{Anne}. But it is restrictive, and it will be seen below that it does not hold in simple examples.

\quad

Let us finally discuss conditions under which RBSDE (\ref{dparametrized}) has a unique solution.
Since we need to refer to these conditions later, let us collect them in the following assumption.

\begin{assum}\label{boundedness-alpha}
 (1) $\xi\in {\mathcal{L}}^{2}$;\\
 (2) $L\in{\mathcal{S}}^{2}$; \\
 (3) $\alpha:\widehat{\Omega}\times[0,T]\rightarrow\R^{+}$ is bounded $\widehat{\p}-a.e.$
\end{assum}

\begin{thm}
Under Assumption (\ref{boundedness-alpha}), there exists a unique solution for RBSDE (\ref{dparametrized}). It coincides with the value function of an American contingent claim with the payoff $L_{t}\alpha_{t}(.)1_{[0,T[}(t)+\xi\alpha_{T}(.)1_{\{T\}}(t)$  and $\tau^{*}(.)=\inf{\{s\in[t,T] : Y_s(.)=L_s\alpha_{s}(.)\}}\wedge{T}$ is
the optimal stopping time for the buyer. Furthermore if Novikov's condition (\ref{novikov}) is satisfied, then RBSDE (\ref{Qrbsde}) has a unique solution.
\end{thm}
\begin{proof}
Under Assumption (\ref{boundedness-alpha}), the integrability conditions (i') and (iii') from Section \ref{shart2} are fulfilled by $\xi(.)=\xi\alpha_T(.)$ and $L(.)=L\alpha(.)$. Thus by Theorem \ref{mthm} and Remark \ref{rem2}, there exists a unique solution for RBSDE (\ref{dparametrized}), and it coincides with the value of
the corresponding optimal stopping problem on the product space. Existence and uniqueness of the solutions of RBSDE (\ref{Qrbsde}) follow from Remark \ref{frem}.
\end{proof}

The following example illustrates that for $t>0$ boundedness of $\frac{\beta_t}{\alpha_t}$ may be easily missed, though $\alpha_t$ is bounded.
\begin{exa}
 Let $G=B_{T}+X$, where $B_{T}$ is the endpoint of a one-dimensional $\F-$Brownian motion with $B_{0}=0$ and $X$ a random variable with centered normal distribution with variance $\epsilon>0$ which is independent of $\F.$
In this case the buyer has noisy information about $B_{T}.$ Due to independence, we know that $G$ has a normal law with mean zero and variance $T+\epsilon$. Therefore we have for all $t\in[0,T]$
 \begin{eqnarray*}
  \p\left(B_{T}+X\in du|\f_{t}\right)&=&\p\left(B_{T}+X-B_{t}+B_{t}\in du|\f_{t}\right)\\
  &=&\p\left(B_{T}+X-B_{t}\in du-y\right)|_{y=B_{t}}\\
  &=&\frac{1}{\sqrt{2\pi(T-t+\epsilon)}}exp\left(-\frac{(u-B_t)^2}{2(T-t+\epsilon)}\right)du\\
  &=&\alpha_{t}(u)\p(B_{T}+X\in du),\\
 \end{eqnarray*}
where $\alpha_{t}(u)=\sqrt{\frac{(T+\epsilon)}{(T-t+\epsilon)}}exp\left(-\frac{(u-B_t)^2}{2(T-t+\epsilon)}+\frac{u^2}{2(T+\epsilon)}\right), u\in\R$. So here the conditional law of $G$ given $\f_{t}$ is absolutely continuous with respect to the law of $G$ for all $t\in[0,T].$
Note that  for all $u\in\R,\ \alpha_{0}({u})=1$, that $\alpha$ is continuous in $(t,u)\in(0,T]\times\R$, and that by $\epsilon>0$
we have
$$\lim_{u\rightarrow\pm\infty}\alpha_{t}(u)=0,\ \ \p-a.s.$$
Therefore, for all $t\in[0,T], \alpha_{t}(\cdot)$ is bounded $\widehat{\p}-a.e.$
It is known from \citep{Imkeller} that $\beta_{t}(.)$ is the Malliavin trace of $\alpha_{t}(.)$. So we have $\frac{\beta_{t}(.)}{\alpha_{t}(.)}=D_{t}\ln(\alpha_{t}(.)),\ \ t\in[0,T].$ Therefore we obtain,
 $\frac{\beta_{t}(u)}{\alpha_{t}(u)}=\frac{1}{(T-t+\epsilon)}(u-B_{t})$ which is not bounded.

\end{exa}


\subsection{American contingent claims with asymmetric information and parametrized RBSDE}

 In this subsection we will rigorously establish the link between optimal solutions for American contingent claims for which the buyer has privileged information and
 solutions of RBSDE w.r.t. enlarged filtrations.

\begin{lem}
 Under the assumptions (i') and (iii') from Section \ref{shart2} we have for $t\in[0,T]$
 $$V^{G}_{t}= \esssup_{\tau^{'}\in\T_{t,T}{(\G)}}\E{\left[R(\tau^{'})|\g_{t}\right]}=\frac{Y_{t}(G)}{\alpha_{t}(G)}=\widehat{Y}_{t}(G),$$
where $Y(.)$ is the solution of the RBSDE (\ref{dparametrized}) and $\widehat{Y}(G)$ satisfies RBSDE (\ref{grbsde}).
 Furthermore, $\tau^{*}:\widehat{\Omega}\rightarrow\R^{+}$ defined by $$\tau^{*}(G)=\inf{\{s\in[t,T] : Y_s(G)=L_s\alpha_s(G)\}}\wedge{T}=\inf{\{s\in[t,T] : \widehat{Y}_s(G)=L_s\}}\wedge{T}$$ is the optimal stopping time for the buyer after time $t$.

\label{lem1}\end{lem}

 \begin{proof}
  Theorem \ref{thm1} gives
  \begin{eqnarray*}
   V^{G}_{t}={\esssup_{\tau^{'}\in\T_{t,T}{(\G)}}\E{\left[R(\tau^{'})|\g_{t}\right]}}={\frac{1}{\alpha_{t}(G)}{\left(\esssup_{\tau(.)\in\T_{t,T}{(\widehat{\F})}}\widehat{\E}\left[R(.,\tau(.))|\widehat{\f}_{t}\right]\right)}_{G}}.
  \end{eqnarray*}
  Remark \ref{rem2} implies
  $$Y_{t}(.)=\esssup_{\tau(.)\in\T_{t,T}{(\widehat{\F})}}\widehat{\E}\left[R(.,\tau(.))|\widehat{\f}_{t}\right].$$

  Furthermore, $$\tau^{*}(.)=\inf{\{s\in[t,T] : Y_s(.)=L_s\alpha_s(.)\}}\wedge{T}$$ is the optimal stopping time. The proof is completed by recalling
  $\widehat{Y}(G)=\frac{Y(G)}{\alpha(G)}$ from the definition in the previous section.

 \end{proof}

\begin{cor}
 The previous lemma implies in particular that $$V^{G}=\esssup_{\tau^{'} \in \mathcal{T}_{0,T}\left(\mathbb{G}\right)}\E\left[ R(\tau^{'})|\g_{0}\right]=Y_{0}(G)=\widehat{Y}_{0}(G),$$ since $\alpha_{0}\equiv 1$. Therefore, the value of the American contingent claim with extra information is given by
 the initial solution of the parametrized RBSDE (\ref{dparametrized}) evaluated at $G$.
\label{value}\end{cor}

\begin{lem}
 Under assumptions (i') and (iii') from Section \ref{shart2} we have for $t\in[0,T]$
 \begin{equation}
  \esssup_{\tau^{'}\in\T_{t,T}{(\G)}}\E{\left[R(\tau^{'})|\f_{t}\right]}=\int_{\R}{Y_{t}(u)dP^{G}(u)}.
\label{equ}\end{equation}

 where $Y(.)$ is the solution of the RBSDE (\ref{dparametrized}) and $\tau^{*}(G)$ the optimal stopping time for the buyer after time $t$.
 \label{lem2}\end{lem}
\begin{proof}
  The proof follows easily from Theorem \ref{thm2} and Remark \ref{rem2}.
 \end{proof}

The following example exhibits a more explicit description of the value of an American call option with additional information.
\begin{exa}
 Consider an American call option with payoff $R(t)=(S_{t}-K)^{+},$ where $K$ is the strike price. The stock price process $S$  satisfies for $t\in[0,T]$
 $$dS_{t}=\mu S_{t}dt+\sigma S_{t}dB_{t},$$ where $\mu$ is the drift, $\sigma>0$ the volatility.
Suppose that $G$ is a random variable such that $\alpha$ is bounded $\p\otimes P^G-a.e$. From Theorem \ref{thm1}, we have for $t\in[0,T]$
 $$V^{G}_{t}={\frac{1}{\alpha_{t}(G)}{\left(\esssup_{\tau(.)\in\T_{t,T}{(\widehat{\F})}}\widehat{\E}\left[(S_{\tau(.)}-K)^{+}\alpha_{\tau(.)}(.)|\widehat{\f}_{t}\right]\right)}_{G}}.$$
 We define $$V_{t}(.):=\esssup_{\tau(.)\in\T_{t,T}{(\widehat{\F})}}\widehat{\E}\left[(S_{\tau(.)}-K)^{+}\alpha_{\tau(.)}(.)|\widehat{\f}_{t}\right],\quad t\in[0,T].$$
From known results about the Snell envelope, we have $\tau^{*}(.)=\inf{\{s\in[t,T] : V_{s}(.)=L_s\alpha_s(.)\}}\wedge{T}$ is optimal in the sense
 $$V^{G}_{t}={\frac{1}{\alpha_{t}(G)}{\left(\widehat{\E}\left[(S_{\tau^{*}(.)}-K)^{+}\alpha_{\tau^{*}(.)}(.)|\widehat{\f}_{t}\right]\right)}_{G}}.$$
Now from Proposition \ref{my proposition}, $$V^{G}_{t}={\frac{1}{\alpha_{t}(G)}}\E\left[(S_{\tau^{*}(u)}-K)^{+}\alpha_{\tau^{*}(u)}(u)|\f_{t}\right]_{u=G},\ \ \p-a.s.$$
The process $S$ is a semimartingale. So from Tanaka's formula the following decomposition for $V^{G}$ is obtained for $t\in[0,T]$:
$$\alpha_t(G)\,\,V^{G}_{t}=(S_{0}-K)^{+}\alpha_{t}(G)+\E[\alpha_{\tau^{*}(u)}(u)\int_{0}^{\tau^{*}(u)}I\{S_{s}>K\}dS_{s}|\f_{t}]_{u=G}+\frac{1}{2}{\E[\alpha_{\tau^*(u)}(u) l^{K}_{\tau^{*}(u)}(S)|\f_{t}]_{u=G}},\ \ \p-a.s.$$
where $l^{K}(S)$ is the local time of $S$ at $K$. Since in particular $\alpha_{0}(G)=1$ and $\f_{0}$ is trivial, we have
$$V^{G}=(S_{0}-K)^{+}+\E[\alpha_{\tau^{*}(u)}(u)\int_{0}^{\tau^{*}(u)}I\{S_{s}>K\}dS_{s}]_{u=G}+\frac{1}{2}{\E[\alpha_{\tau^{*}(u)}(u)\,l^{K}_{\tau^{*}(u)}(S)]_{u=G}},\ \ \p-a.s.$$
On the other hand $S_{t}=S_{0}e^{\sigma B_{t}+(\mu-\frac{1}{2}{\sigma}^2)t}, t\in[0,T].$ Therefore $L=(S-K)^{+}$, $\xi=(S_{T}-K)^{+}$ and $\alpha$ satisfies Assumption
\ref{boundedness-alpha} since $e^{\sigma B}$ is a continuous function and $\E(e^{\sigma B_{t}})=e^{\frac{1}{2}t^{2}\sigma^{2}}<\infty$ for each $t\in[0,T].$ Hence Lemma \ref{lem1}
provides a representation for the solution of RBSDE (\ref{grbsde}) with barrier $(S-K)^{+}$ and final value $(S_{T}-K)^{+}$, where $S$ is a geometric Brownian motion.
\end{exa}

 \begin{rem}
 We have $\widehat{Y}(.)=\frac{Y(.)}{\alpha(.)}$, so if we replace $Y(\cdot)$ by $\widehat{Y}(\cdot)\alpha(\cdot)$ in (\ref{equ}), we get
 $$ \esssup_{\tau^{'}\in\T_{t,T}{(\G)}}\E{\left[R(\tau^{'})|\f_{t}\right]}=\int_{\R}{\widehat{Y}_{t}(u)\alpha_{t}(u)dP^{G}(u)}=\E[\widehat{Y}_{t}(G)|\f_{t}].$$
 where $\widehat{Y}(G)$ solves the RBSDE (\ref{grbsde}) in the initially enlarged filtration. The last equation is due to the definition of $\alpha(\cdot)$.
 \label{rem3}\end{rem}

\begin{rem}
 Under Assumption (\ref{boundedness-alpha}), RBSDE (\ref{dparametrized}) has a unique solution with first component $Y(.)$. Furthermore, RBSDE (\ref{grbsde}) has a unique solution whose first component
 coincides with $V^{G}$. On the other hand, from Lemma \ref{lem1} we have $V^{G}=\widehat{Y}(G),\ \p-a.s.$ Thus $\widehat{Y}(G)$ is the unique solution of RBSDE (\ref{grbsde}). 
  However, square integrability of other components of the solution remains open. They are not necessarily unique, being derived from the Doob-Meyer decomposition for continuous supermartingales, as shown in \citep{Nichole}.
\end{rem}

\section{Cost of additional information}\label{sec4}
For American contingent claims, the buyer has to select a stopping time $\tau\in \T_{0,T}$ at which he exercises his option in such a way that the expected payoff $R(\tau)$ is maximized. If he has privileged information, he has access to a larger set of exercise times leading to a higher expected payoff. The value of the additional information
can be interpreted as the price he should pay to obtain it. From a utility indifference point of view, the price should be defined as the difference of the maximal expected payoff the buyer receives with additional information and the maximal expected payoff without.

\subsection{Definition and primary results}

To investigate this value in our framework. We denote the cost of the extra information with $CEI$, and define more formally
\begin{defn}
$$CEI(t):=\esssup_{\tau^{'}\in\T_{t,T}{(\G)}}\E{\left[R(\tau^{'})|\g_{t}\right]}-\esssup_{\tau\in\T_{t,T}{(\F)}}\E{\left[R(\tau)|\f_{t}\right]},\quad t\in[0,T],$$
and
\begin{eqnarray*}
CEI:= CEI(0) = \esssup_{\tau^{'}\in\T_{0,T}{(\G})}\E{\left[R(\tau^{'})|\sigma(G)\right]}-\sup_{\tau\in\T_{0,T}{(\F)}}\E{\left[R(\tau)\right]}.
\end{eqnarray*}

The last equation follows from the triviality of $\f_{0}$ and $\g_{0}=\sigma(G)$ ( see Remark \ref{r2.1}). We call
$CEI(\cdot)$ the \emph{value function} of the additional information.
\end{defn}

We have for $t\in[0,T]$
$$CEI(t)=\left(\esssup_{\tau^{'}\in\T_{t,T}{(\G)}}\E{\left[R(\tau^{'})|\g_{t}\right]}-\esssup_{\tau^{'}\in\T_{t,T}{(\G)}}\E{\left[R(\tau^{'})|\f_{t}\right]}\right)+\left(\esssup_{\tau^{'}\in\T_{t,T}{(\G)}}\E{\left[R(\tau^{'})|\f_{t}\right]}-\esssup_{\tau\in\T_{t,T}{(\F)}}\E{\left[R(\tau)|\f_{t}\right]}\right).$$
The second expression is a non-negative random variable. We prove that the expectation of the first expression is also positive and thus $\E[CEI(t)]$ is a positive quantity. 
By the tower property of conditional expectation we have
\begin{eqnarray*}
 \esssup_{\tau^{'}\in\T_{t,T}{(\G)}}\E{\left[R(\tau^{'})|\f_{t}\right]}=\esssup_{\tau^{'}\in\T_{t,T}{(\G)}}\E{\left[\E{\left[R(\tau^{'})|\g_{t}\right]}|\f_{t}\right]}\leq \esssup_{\tau^{'}\in\T_{t,T}{(\G)}}\E{\left[V_{t}^{G}|\f_{t}\right]}={\E\left[V_{t}^{G}|\f_{t}\right]},\quad\p-a.s.
\end{eqnarray*}
Therefore we obtain that $\E[CEI(t)]\geq 0$ for $t\geq 0$.

\quad

If we suppose again that $\F$ is a Brownian filtration as in section 4, we are able to link $CEI(t)$ to RBSDE as follows:
\begin{cor}
 Under Assumption (\ref{boundedness-alpha}), Lemma \ref{lem1} and Remark \ref{rem1}  yield the equation
\begin{equation}
 CEI(t)=\frac{Y_{t}(G)}{\alpha_{t}(G)}-Y_{t}=\widehat{Y}_{t}(G)-Y_{t},\quad t\in[0,T],
\label{CEIt}\end{equation}\\
where $Y(.)$ is the solution of (\ref{dparametrized}), $\widehat{Y}(G)$ the solution of (\ref{grbsde}), and $Y$ is the solution of the RBSDE

\begin{equation}
\left\{
\begin{aligned}
-dY_{t}&=dK_{t}-Z_{t}dB^{\F}_{t},\quad\quad 0\leq t \leq T,\\
Y_{T}&=\xi,\\
Y_{t}&\geq L_{t},\quad 0\leq t \leq T,\quad\quad\textstyle\int_{0}^{T}{(Y_{t}-L_{t})dK_{t}}=0.
 \end{aligned}
 \right.
\end{equation}
Since in particular $\alpha_{0}(G)=1$, we can express $CEI$ as the difference of the initial values of solutions of two RBSDE, namely
\begin{equation}
 CEI=Y_{0}(G)-Y_{0}=\widehat{Y}_{0}(G)-Y_{0}.
\label{CEI}\end{equation}
\end{cor}

\begin{rem}
From the remarks preceding the above corollary, we conclude that $\E[\widehat{Y}_{t}(G)]\geq \E[Y_{t}]$ for $t\geq 0$. In other words, the average of the solution of the initially enlarged RBSDE is bigger than the average of the solution of the initial RBSDE.
\end{rem}

Let us briefly comment on $CEI(T)$, the value of extra information at exercise time $T$ from the perspective of the RBSDE. By definition we have $$CEI(T):=\esssup_{\tau^{'}\in\T_{T,T}{(\G)}}\E{\left[R(\tau^{'})|\g_{T}\right]}-\esssup_{\tau\in\T_{T,T}{(\F)}}\E{\left[R(\tau)|\f_{T}\right]}=\E{\left[R(T)|\g_{T}\right]}-\E{\left[R(T)|\f_{T}\right]}=\xi-\xi=0.$$
Looking at this value with the underlying RBSDE, we get (see \ref{CEIt})
 $$CEI(T)=\frac{Y_{T}(G)}{\alpha_{T}(G)}-Y_{T}=\widehat{Y}_{T}(G)-Y_{T}.$$
 But $\frac{Y_{T}(G)}{\alpha_{T}(G)}=\frac{\xi\alpha_{T}(G)}{\alpha_{T}(G)}=\xi$, and $Y_{T}=\widehat{Y}_{T}(G)=\xi$, which confirms $CEI(T)=0.$
This is what we expect, since additional information at
exercise time does not help the buyer to do better by a better strategy. It would be interesting to find a more precise description of the price of the additional information. As it stands, it is given by the difference of the first components $Y$ of two solution processes of RBSDE with identical terminal conditions, drivers, and obstacles, but on two spaces of different complexity. We conjecture that $Y$ is an increasing function of the complexity of the spaces, but at the moment cannot substantiate this claim.


\subsection{A special case}

We briefly discuss a simple case for which $CEI$ can be explicitly calculated. Assume that  $\F=(\f_{t})_{t\in[0,T]}$ is a Brownian standard filtration and $G$ is independent of $\f_{t}$ for all $t\in[0,T]$.
In this case we have for $t\in[0,T], u\in\R$
$$\alpha_{t}(u) =\frac{dP_{t}^{G}(u,\cdot)}{dP^{G}(u)}=1,\quad \p- a.s,$$ so from formula (\ref{CEI})
$CEI=0$. This is because we face the RBSDE

\begin{equation}
\left\{
\begin{aligned}
-dY_{t}(.)&=dK_{t}(.)-Z_{t}(.)dB^{\F}_{t},\quad\ \quad 0\leq t \leq T,\\
Y_{T}(.)&=\xi\alpha_{T}(.)=\xi,\\
Y_{t}(.)&\geq L_{t}\alpha_{t}(.)=L_{t},\quad0\leq t \leq T,\quad\quad\textstyle\int_{0}^{T}{(Y_{t}(.)-L_{t}\alpha_{t}(.))dK_{t}(.)}=\textstyle\int_{0}^{T}{(Y_{t}(.)-L_{t})dK_{t}(.)}=0.\\
\end{aligned}
\right.
\label{akhari}\end{equation}

By uniqueness of the solution of the RBSDE, $Y(.)\equiv Y$.\\
 In addition, $V^{G}$, the value of the American contingent claim with additional information coincides with the value of the same American contingent claim
 without this information. This follows from Remark \ref{value} stating $V^{G}=Y_{0}(G),$ where $Y(.)$ is the solution of (\ref{akhari}), and uniqueness of its solution giving $Y(G)=Y.$

\quad

\vspace{1cm}
{\Large{\bf Acknowledgements}}

\quad

{\bf We would like to thank the referees for their careful reading and helpful comments.}
This work is a part of the first author's PhD thesis at Sharif University of Technology under supervision of Professor Bijan Z.Zangeneh. She wishes to thank her supervisor for his support, encouragement and guidance. Furthermore, she wants to thank Viktor Feunou, PhD candidate at Humboldt-Universit\"at zu Berlin, for helpful discussions. The financial support from Humboldt-Universit\"at zu Berlin is gratefully acknowledged.
In particular, she wants to thank Professor Peter Imkeller for making this possible.


\newpage
\begin{thebibliography}{1}
\bibitem{Amen01} Amendinger, J. (1999). Initial Enlargement of Filtrations and Additional Information in Financial Markets, \textit{PhD thesis, Technischen Universit\"at zu Berlin}.
\bibitem{Amendinger} Amendinger, J., Imkeller, P. and Schweizer, M. (1998). Additional logarithmic utility of an insider, \textit{Stochastic Process. Appl.}, \textbf{75}, 263--268.
\bibitem{Ankirchner} Ankirchner, S., Dereich, S. and Imkeller, P. (2006). The Shannon information of filtrations and the additional logarithmic utility of insiders. \textit{Ann. Probab.}, \textbf{34}, 743--778.
\bibitem{Bismut} Bismut, J.M. (1973). Conjugate convex functions in optimal stochastic control, \textit{J. Math. Anal. Appl.}, \textbf{44}, 384--404.
\bibitem{Black} Black, F.,  Scholes, M. (1973). The pricing of options and corporate liabilities, \textit{J. Polit. Econ.}, \textbf{81}, 637--654.
\bibitem{Bouch} Bouchaud, J-P, Sornette, D. (1994). The Black-Scholes option pricing problem in mathematical finance: generalization and extensions for a large class of stochastic processes, \textit{J. Phys. I (France)}, \textbf{4}, 863--881.
\bibitem{Zargar} Callegaro, G., Jeanblanc, M. and Zargari, B. (2013). Carthaginian enlargement of filtrations, \textit{ESAIM Probab. Stat.}, \textbf{17}, 550--566.
\bibitem{Duffi} Duffie, D. (1988). \textit{Security markets: Stochastic models}, Academic Press: Boston.
\bibitem{Nichole} El Karoui, N., Kapoudjian, C., Pardoux, E., Peng, S. and  Quenez, M.C. (1997). Reflected solutions of backward SDE’s, and related obstacle problems for PDE’s, \textit{Ann. Probab.}, \textbf{25}, 702--737.
\bibitem{Nichole2} El Karoui, N., Quenez, M.C. (1995) , Dynamic Programming and Pricing of Contingent Claims in an Incomplete Market, \textit{SIAM J. Control Optim.}, \textbf{33}, 29-66.
\bibitem{Elli} Elliott, R.J., Geman, H. and Korkie, B.M. (1997). Portfolio optimization and contingent claim pricing with differential information, \textit{Stochastics Stochastics Rep.}, \textbf{60}, 185--203.
\bibitem{Anne} Eyraud-Loisel, A. (2005).  Backward stochastic differential equations with enlarged filtration: Option helging of an insider trader in a financial market with jumps, \textit{Stochastic Process. Appl.}, \textbf{115}, 1745--1763.
\bibitem{Foll}F\"{o}llmer, H., Schweizer, M. (1990). \textit{Hedging of contingent claims under incomplete information}, in Applied Stochastic Analysis (eds.) M. H. A. Davis and R. J. Elliott, Gordon and Breach: London. 
\bibitem{Follmer} F\"{o}llmer, H., Sondermann, D. (1986). Hedging of non-redundant contingent claims, \textit{In W. Hildenbrand and A. Mas-Collel (eds.), Contributions to Mathematical Economics}, 205--223.
\bibitem{grigorova} Grigorova, M., Imkeller, P., Offen, E., Ouknine, Y. (2016). Reflected BSDEs when the obstacle is not right-continuous and optimal stopping. arXiv:1504.06094.
\bibitem{Gror1} Grorud, A., Pontier, M. (1998). Asymmetric information and incomplete market, \textit{Int. J. Theor. Appl. Finance}, \textbf{4}, 285--302.
\bibitem{Gror2} Grorud, A., Pontier, M. (1998). Insider trading in a continuous time market model, \textit{Int. J. Theor. Appl. Finance}, \textbf{1}, 331--347.
\bibitem{Hamedane} Hamad\a'{e}ne, S. (2002). Reflected BSDE with discontinuous  barriers, \textit{Stochastics Stochastics Rep.}, \textbf{74}, 571--596.
\bibitem{Ham1} Hamad\a'{e}ne, S., Lepeltier, J.P. (2000). Reflected BSDE's and mixed game problem, \textit{Stochastic Processes Appl.}, \textbf{85}, 177--188.
\bibitem{Harri1}Harrison, M., Kreps, D. (1979). Martingale and arbitrage in multiperiod securiries markets, \textit{Int. J. Econ. Theory}, \textbf{20}, 381--408.
\bibitem{Harri2} Harrison, M.,  Pliska, S.R. (1981). Martingales and stochastic integrals in the theory of continuous trading, \textit{Stochastic Processes Appl.}, \textbf{11}, 215--260.
\bibitem{Imkeller} Imkeller, P. (2003). Malliavin's calculus in insider models: additional utility and free lunches, \textit{Math. Finance}, \textbf{13}, 153--169.
\bibitem{Nick} Imkeller, P., Perkowski, N. (2015). The existence of dominating local martingale measures, \textit{Finance Stoch.}, \textbf{19}, 685--717.
\bibitem{jacod1} Jacod, J. (1979). \textit{Calcul stochastique et probl\a'{e}mes de martingales}, Lecture Notes in Mathematics 714. Springer-Verlag: Berlin.
\bibitem{jacod2} Jacod, J. (1985). \textit{Grossissement Initial, Hypoth\a'{e}se(H') et Th\a`{e}or\a'{e}me de Girsanov}, Lecture Notes in Mathematics 1118. Springer-Verlag: Berlin.
\bibitem{Kara} Karatzas, I. (1989). Optimization problems in the theory of continuous trading, \textit{SIAM. J. Control Optim.}, \textbf{27}, 1221--1259.
\bibitem{Idris} Kharroubi, I., Lim, T. (2014). Progressive enlargement of filtrations and backward stochastic differential equations with jumps, \textit{J. Theoret. Probab.}, \textbf{27}, 683--724.
\bibitem{Kob} Kobylanski, M., Quenez, M.C. (2012). Optimal stopping time problem in a general framework, \textit{Electron. J. Probab.}, \textbf{17}, 1--28.
\bibitem{Lepeltier}  Lepeltier, J.P., Xu, M. (2005). Penalization method for reflected backward stochastic differential equations with one r.c.l.l. barrier,
\textit{Statist. Probab. Lett.}, \textbf{75}, 58--66.
\bibitem{Merton1} Merton, R. (1973). Theory of rational option pricing, \textit{Bell J. Econ. Manage. Sci.}, \textbf{4}, 141--183.
\bibitem{Merton2} Merton, R. (1991). \textit{ Continuous time finance}, Basil Blackwell: Oxford.
\bibitem{Muller} M\"uller, S. (1985). \textit{Arbitrage pricing of contingent claims}, Lecture Notes in Economics and Mathematical Systems,  254, Springer-Verlag: Berlin.
\bibitem{Nvu} Neveu, J. (1975). \textit{Discrete-parameter martingales}, North-Holland: Amsterdam.
\bibitem{Oks} Oksendal, B. (2003). \textit{Stochastic differential equations: An introduction with applications}, Springer-Verlag: Berlin.
\bibitem{Pardoux} Pardoux, E., Peng, S. (1990). Adapted solution of a backward stochastic differential equation, \textit{Systems Control Lett.}, \textbf{14}, 55--61.
\bibitem{Pik} Pikovsky, I., Karatzas, I. (1996). Anticipative portfolio optimization, \textit{Adv. in Appl. Probab.}, \textbf{28}, 1095--1122.
\bibitem{Schal} Sch\"al, M. (1994). On quadratic cost criteria for option hedging, \textit{Math. Oper. Res.}, \textbf{19}, 121--131.
\bibitem{Sch} Schweizer, M. (1995). Variance-optimal hedging in discrete time, \textit{Math. Oper. Res.}, \textbf{20}, 1--32.
\bibitem{yo} Stricker, C., Yor, M. (1978). Calcul stochastique d\a'{e}pendant d'un param\^{e}tre, \textit{Z. W. Geb.}, \textbf{45}, 109--133.
\end{thebibliography}
 \end{document}